\documentclass[11pt,twoside]{article}
\usepackage{url}
\usepackage{xspace}
\usepackage{latexsym}
\usepackage{epsfig}
\usepackage{fancybox}
\usepackage{amssymb}
\usepackage{amsmath}
\usepackage{amsfonts}
\usepackage{rotating}
\usepackage{booktabs}
\usepackage{amsthm,stmaryrd}
\usepackage{indentfirst}
\usepackage{color}
\usepackage{geometry}
\usepackage{tikz}
\usepackage{graphicx}
\usepackage{wrapfig}
\usepackage{fancyhdr}
\usepackage{hyperref}
\hypersetup{
    colorlinks=true,
    citecolor=blue,
    linkcolor=red,
    filecolor=magenta,
    urlcolor=cyan,
}

\newtheorem{theorem}{Theorem}[section]

\newtheorem{corollary}[theorem]{Corollary}
\newtheorem{lemma}[theorem]{Lemma}

\newtheorem{definition}[theorem]{Definition}
\newtheorem{example}[theorem]{Example}

\newtheorem{proposition}[theorem]{Proposition}

\numberwithin{equation}{section}

\let\set\mathbb
\def\lc{\operatorname{lc}}

\def\shift{\text{S}}
\def\reduction{\operatorname{red}}
\def\cont{\operatorname{cont}}
\def\prim{\operatorname{prim}}

\title{Constructing Minimal Telescopers for Rational Functions\\ in
  Three Discrete Variables}
\author{Shaoshi Chen$^1$, Qing-Hu Hou$^2$, Hui Huang$^3\footnote{Corresponding author.}$,
  George Labahn$^4$, Rong-Hua Wang$^5$
  }
\date{\small
  $^1$ KLMM, Academy of Mathematics and Systems Science, Chinese Academy of Sciences,\\
  Beijing, 100190, China\\
  $^2$ Center for Applied Mathematics, Tianjin University,
  Tianjin, 300072, China\\
  $^3$ School of Mathematical Sciences, Dalian University of Technology,
  Dalian, Liaoning, 116024, China\\
  $^4$  Cheriton School of Computer Science, University of Waterloo,
  Waterloo, ON, N2L 3G1, Canada\\
  $^5$ School of Mathematical Sciences, Tiangong University,
  Tianjin, 300387, China\\
  {\sf schen@amss.ac.cn}, {\sf qh$\_$hou@tju.edu.cn},
  {\sf huanghui@dlut.edu.cn}\\
  {\sf glabahn@uwaterloo.ca},
  {\sf wangronghua@tiangong.edu.cn}
}

\begin{document}
\maketitle

\begin{center}
Dedicated to Professor Ziming Li on the occasion of his 60th birthday.
\end{center}

\bigskip
\begin{abstract}
  We present a new algorithm for constructing minimal telescopers for
  rational functions in three discrete variables. This is the first
  discrete reduction-based algorithm that goes beyond the bivariate
  case. The termination of the algorithm is guaranteed by a known
  existence criterion of telescopers. Our approach has the important
  feature that it avoids the potentially costly computation of
  certificates. Computational experiments are also provided so as to
  illustrate the efficiency of our approach.
  
  
\end{abstract}

\section{Introduction}\label{SEC:intro}
Creative telescoping \cite{Zeil1990b,Zeil1991} is a powerful tool used to find closed form
solutions for definite sums and definite integrals.  The method
constructs a recurrence (resp.\ differential) equation satisfied by
the definite sum (resp.\ integral) with closed form solutions over a
specified domain resulting in formulas for the sum or integral.
Methods for finding such closed form solutions are available for many
special functions, with examples given
in~\cite{Abra1989,Petk1992,Abra1995b,APP1998,BaPe1999,ClvH2006,CvHL2010,vHoe2017,ABPS2021}.
Even when no closed form exists the method of creative  telescoping often remains
useful. For example the resulting recurrence or differential equation
enables one to determine asymptotic expansions and derive other
interesting facts about the original sum or integral.

In the case of summation, specialized to the trivariate case, in order
to compute a sum of the form
$$\sum_{y=a_1}^{b_1}\sum_{z=a_2}^{b_2}f(x,y,z),$$ the main task of
creative telescoping consists in finding $c_0,\dots, c_\rho$, rational
functions (or polynomials) in $x$, not all zero, and two functions
$g(x,y,z),h(x,y,z)$ in the same class of functions as $f(x,y,z)$ such that
\begin{equation}\label{EQ:ct}
c_0 f+ c_1 \shift_x(f) + \cdots + c_\rho\shift_x^\rho(f)
= (\shift_{y}(g)-g) +(\shift_{z}(h)-h),
\end{equation}
where $\shift_x$, $\shift_y$ and $\shift_z$ denote shift operators in
$x$, $y$ and $z$, respectively. The number $\rho$ may or may not be
part of the input.  If $c_0,c_1,\dots, c_\rho$ and $g,h$ are as above, then
$L = c_0 + c_1\shift_x+\cdots + c_\rho\shift_x^\rho$ is called a {\em telescoper}
for $f$ and $(g,h)$ is a {\em certificate} for $L$.

The utility of creative telescoping is best demonstrated by examples.
Suppose we want to find a closed form of the
following multiple sum
\[
\sum_{y=0}^{x}\sum_{z=0}^{x} f(x,y,z)
\quad \text{with}\ f(x,y,z) = \frac{2y-x}{(x+y+1)(-2x+y-1)(x+z+1)}.
\]
To this end, the method first constructs a telescoper
$L = \shift_x{ }-1$ for $f$ and a corresponding certificate
\[
(g, h) =
\left(
\frac{8x^2{-}2xy{-}y^2{+}19x{-}2y{+}11}
     {(x{+}y{+}1)({-}2x{+}y{-}3)({-}2x{+}y{-}2)(x{+}z{+}1)},
\frac{{-}x{+}2y{-}1}{(x{+}y{+}2)({-}2x{+}y{-}3)(x{+}z{+}1)}
\right)
\]
such that
\begin{equation*}
L(f) = f(x+1,y,z)-f(x,y,z) = g(x,y+1,z)-g(x,y,z)+h(x,y,z+1)-h(x,y,z).
\end{equation*}
Summing on both sides over $y,z$ from zero to $x$, and applying the
idea of telescoping to $g$ for $y$ and to $h$ for $z$, respectively,
yield
\begin{align*}
  &\quad\,
  \sum_{y=0}^x\sum_{z=0}^x f(x+1,y,z)-\sum_{y=0}^x\sum_{z=0}^x f(x,y,z)\\
  &= \sum_{z=0}^x\big(g(x,x+1,z)-g(x,0,z)\big)
  + \sum_{y=0}^x\big(h(x,y,x+1)-h(x,y,0)\big).
\end{align*}
Employing the notation $F(x) = \sum_{y=0}^{x}\sum_{z=0}^{x}f(x,y,z)$,
along with a range match-up, one obtains
\begin{align}
F(x+1) - F(x) &= \sum_{z=0}^x\big(g(x,x+1,z)-g(x,0,z)+f(x+1,x+1,z)\big)
\nonumber \\
&+ \sum_{y=0}^x\big(h(x,y,x+1)-h(x,y,0)+f(x+1,y,x+1)\big)
+ f(x+1,x+1,x+1)
\nonumber\\
&=\sum_{z=0}^x\frac{x+1}{(x+2)(2x+3)(x+z+1)(x+z+2)}
\nonumber\\
&+\sum_{y=0}^x\frac{x-2y+1}{2(x+1)(2x+3)(x+y+2)(-2x+y-3)}
-\frac{x + 1}{(x + 2)(2x + 3)^2},
\label{EQ:ctrec}
\end{align}
where the right-hand side merely involves single sums and thus the
problem is now reduced to finding closed forms of these sums. Applying
the method of creative telescoping (specialized to the bivariate case)
again, one finds that the first single sum is equal to $1/(2(x+2)(2x+3))$,
while the second sum admits a first-order linear recurrence equation,
which yields the closed form $-1/(2(x + 2)(2x + 3)^2)$. A direct
calculation confirms that the right-hand side of \eqref{EQ:ctrec}
collapses to zero after expansion, that is, $F(x+1) - F(x)=0$.
Together with the initial value $F(0) = 0$, one then concludes that
$\sum_{y=0}^{x}\sum_{z=0}^{x} f(x,y,z) = 0$.

Over the past two decades, a number of generalizations and refinements
of creative telescoping have been developed. At the present time
the reduction-based approach has gained high support as it is both
efficient in practice and has the important feature of being able to find a
telescoper for a given function without necessarily computing a corresponding
certificate. This is desirable in a typical situation where only the
telescoper is of interest and its size is much smaller than the size
of the certificate. Even when a certificate is needed, the approach also
allows one to express it as an unnormalized sum so that the summands
are concatenated symbolically without actually calculating the sum.
Such an expression can be more easily specialized at end points of the
summation range than the expanded certificate, and thus turns out to be
useful in many applications.

The reduction-based approach was first developed in the differential
case for bivariate rational functions~\cite{BCCL2010}, and later
generalized to rational functions in several variables~\cite{BLS2013},
to hyperexponential functions~\cite{BCCLX2013}, to algebraic
functions~\cite{CKK2016} and to D-finite
functions~\cite{CvHKK2018,vdHo2021,BCLS2018}.  In the shift case a
reduction-based approach was developed for hypergeometric
terms~\cite{CHKL2015,Huan2016} and to multiple binomial
sums~\cite{BLS2017} (a subclass of the sums of hypergeometric
terms).

In the case of discrete functions having more than two variables no
complete reduction-based creative telescoping algorithm has been known
so far.  Having such an algorithm would allow us to tackle many
multiple summations from applications more efficiently. However, it is
quite challenging to develop an algorithm once for all. As a first step,
in the present paper we address the most fundamental case, namely when
$f, g, h$ in \eqref{EQ:ct} are all rational functions in $x,y,z$.
This is also a natural follow up to the recent work~\cite{CHLW2016, CDZ2019, CDWZ2021}
on the existence problem of telescopers for rational
functions in three variables.

The basic idea of the general reduction-based approach, formulated for
the shift trivariate rational case, is as follows. Let $\set K$ be a
field of characteristic zero. Assume that there is a
$\set K(x)$-linear map $\reduction(\cdot): \set K(x,y,z)\rightarrow \set
K(x,y,z)$ with the property that for all $f\in \set K(x,y,z)$, there
exist $g,h\in \set K(x,y,z)$ such that
$f-\reduction(f)=(\shift_y(g)-g) + (\shift_z(h)-h)$, that is,
$f-\reduction(f)$ is summable with respect to~$y,z$,
and $\reduction(f)$ is minimal in certain sense. In other words,
$\reduction(f)$ indicates the \lq\lq minimum\rq\rq\ adjustments needed
for $f$ to become summable with respect to~$y,z$, which apparently
excludes the most trivial case of $\reduction(f)=f$. Such a map is
called a {\em reduction} with $\reduction(f)$ considered as a
{\em remainder} of $f$ with respect to the reduction
$\reduction(\cdot)$. Then in order to find a telescoper for~$f$, we
can iteratively compute $\reduction(f), \reduction(\shift_x(f)),
\reduction(\shift_x^2(f)), \dots$ until we find a nontrivial linear
dependence over $\set K(x)$.  Once we have such a dependence,
say
\[
c_0\reduction(f) + \cdots + c_\rho\reduction(\shift_x^\rho(f))=0
\]
for $c_i \in \set K(x)$ not all zero, then by linearity,
$\reduction(c_0f+ \cdots + c_\rho \shift_x^\rho(f))=0$, that is, $c_0f
+ \cdots + c_\rho\shift_x^\rho(f) = (\shift_y(g)-g) + (\shift_z(h)-h)$
for some $g,h\in \set K(x,y,z)$. This yields a telescoper $c_0+\cdots
+ c_\rho\shift_x^\rho$ for $f$.

To guarantee the termination of the above process, one possible way is
to show that, for every summable function $f$, we have $\reduction(f)=0$.
If this is the case and $L = c_0 + \cdots + c_\rho \shift_x^\rho$ is a
telescoper for $f$, then $L(f)$ is summable by the definition.  So
$\reduction(L(f)) = 0$, and again by the linearity, $\reduction(f),
\dots, \reduction(\shift_x^\rho(f))$ are linearly dependent over $\set K(x)$.
This means that we will not miss any telescoper and that the method
will terminate provided that a telescoper is known to
exist.  This approach was taken in \cite{CHKL2015}.  It requires us to
know exactly under what kind of conditions a telescoper exists,
so-called the {\em existence problem of telescopers}, and, when these
conditions are fulfilled, then it is guaranteed to find one of minimal
order $\rho$. Such existence problems have been well studied in the case
of bivariate hypergeometric terms~\cite{Abra2003} and more recently in the
trivariate rational case~\cite{CHLW2016, CDZ2019, CDWZ2021}.

A second, alternate way to ensure termination, used for example in
\cite{BCCL2010,BCCLX2013}, is to show that, for a given function $f$,
the remainders $\reduction(f), \reduction(\shift_x(f)),
\reduction(\shift_x^2(f)),\dots$ form a finite-dimensional
$\set K(x)$-vector space. Then, as soon as $\rho$ exceeds this finite
dimension, one can be sure that a telescoper of order at most $\rho$
will be found.  This also implies that every bound for the dimension
gives rise to an upper bound for the minimal order of
telescopers. This approach provides an independent proof for the
existence of a telescoper. However, since such an upper order bound is
only of theoretical interest and will not affect the practical
efficiency of the algorithms, in this paper we will confine ourselves
with the first approach for termination and leave the second approach
for future research.

Our starting point is thus to find a suitable reduction for trivariate
rational functions. In particular we present a reduction
$\reduction(\cdot)$ which satisfies the following properties: (i)
$\reduction(f) = 0$ whenever $f\in \set K(x,y,z)$ is summable and (ii)
$\reduction(f)$ is minimal in certain sense.  One issue with this
reduction, similar to that encountered in the bivariate hypergeometric
case \cite{CHKL2015}, is the difficulty that $\reduction(\cdot)$ is
not a $\set K(x)$-linear map in
general. To overcome this we follow the ideas of \cite{CHKL2015}.
Namely, we introduce the idea of congruences modulo summable rational
functions and show that $\reduction(\cdot)$ becomes $\set K(x)$-linear when
it is viewed as a residue class.  Using the existence criterion of
telescopers established in~\cite{CHLW2016}, we are then able to design
a creative telescoping algorithm from $\reduction(\cdot)$ as described
in the previous paragraphs.

The remainder of the paper proceeds as follows. The next section gives
some preliminary materials needed for this paper, particularly a review
of a reduction method due to Abramov. In Section~\ref{SEC:biabred} we
extend Abramov's reduction method to the bivariate case by
incorporating a primary reduction. In Section~\ref{SEC:remlinearity}
we show that the reduction remainders introduced in the previous
section are well-behaved with respect to taking linear combinations,
followed in Section~\ref{SEC:rct} by a new algorithm for constructing
telescopers for trivariate rational functions based on the bivariate
extension of Abramov's reduction method. In Section~\ref{SEC:test} we
provide some experimental tests of our new algorithm. The paper ends
with some topics for future research.

\section{Preliminaries}\label{SEC:prel}
Throughout the paper we let $\set K$ denote a field of characteristic
zero, with $\set F = \set K(x)$ and $\set F(y,z)$ being the field of
rational functions in $y,z$ over~$\set F$.
Choosing the pure lexicographic order $y\prec z$, we say that a
polynomial in $\set F[y,z]$ is {\em monic} if its highest term
with respect to~$y,z$ has coefficient one. For a nonzero polynomial
$p\in \set F[y,z]$, its degree and leading coefficient with respect to
the variable $v\in \{y,z\}$ are denoted by $\deg_v(p)$ and $\lc_v(p)$,
respectively. We will follow the convention that $\deg_v(0)=-\infty$.

We let $\sigma_y$ and
$\sigma_z$ be the automorphisms over $\set F(y,z)$, which, for any
$f\in \set F(y,z)$, are defined by
\[\sigma_y(f(x,y,z))=f(x,y+1,z) \quad \text{and} \quad
\sigma_z(f(x,y,z))=f(x,y,z+1).\]
Let $G=\langle \sigma_y,\sigma_z\rangle$ be the free abelian
multiplicative group generated by $\sigma_y, \sigma_z$.
The application of an element $\tau = \sigma_y^\alpha\sigma_z^\beta$
in $G$ to a rational function $f\in \set F(y,z)$ is defined as
\[
\tau(f(x,y,z)) = \sigma_y^\alpha\sigma_z^\beta(f(x,y,z))
= f(x,y+\alpha,z+\beta).
\]
For any $\tau\in G$, we say that a polynomial $p\in
\set F[y,z]$ is {\em $\tau$-free} if $\gcd(p,\tau^\ell(p))=1$
for all nonzero $\ell\in \set Z$. A rational function $f\in \set
F(y,z)$ is called {\em $\tau$-summable} if $f=\tau(g)-g$ for some
$g\in \set F(y,z)$. The {\em $\tau$-summability problem} is then to
decide whether a given rational function in $\set F(y,z)$ is
$\tau$-summable or not. Rather than merely giving a negative answer in
case the function is not $\tau$-summable, one could instead seek
solutions for a more general problem, namely the
{\em $\tau$-decomposition problem}, with the intent to make the
nonsummable part as small as possible. Precisely speaking, the
$\tau$-decomposition problem, for a given rational function $f\in \set
F(y,z)$, asks for an additive decomposition of the form $f =
\tau(g)-g+r$, where $g,r\in \set F(y,z)$ and $r$ is minimal in certain
sense such that $f$ would be $\tau$-summable if and only if $r=0$. It
is readily seen that any solution to the decomposition problem tackles
the corresponding summability problem as well.

In the case where $\tau=\sigma_y$, the decomposition problem was first
solved by Abramov in~\cite{Abra1975} with refined algorithms in
\cite{Abra1995a,Paul1995,AbPe2001a,GGSZ2003,Schn2007}. All these algorithms can
be viewed as discrete analogues of the Ostrogradsky-Hermite reduction
for rational integration (and beyond).  We refer to any of these algorithms
restricted to the rational case as the
Abramov reduction.

\begin{theorem}[Abramov reduction]\label{THM:abred}
  Let $f$ be a rational function in $\set F(y,z)$. Then the Abramov
  reduction finds $g\in \set F(y,z)$ and $a,b\in \set F[y,z]$ with
  $\deg_y(a)<\deg_y(b)$ and $b$ being $\sigma_y$-free such that
  \begin{equation*}\label{EQ:abdecomp}
  f = \sigma_y(g)-g + \frac{a}{b}.
  \end{equation*}
  Moreover, if $f$ admits such a decomposition then
  \begin{itemize}
  \item $f$ is $\sigma_y$-summable if and only if $a=0$;
  \item $b$ has the lowest possible degree in $y$ when $\gcd(a,b)=1$. That is,
    if there exist a second $g'\in \set F(y,z)$ and
    $a',b'\in \set F[y,z]$ such that $f = \sigma_y(g')-g'+a'/b'$, then
    $\deg_y(b')\geq \deg_y(b)$.
  \end{itemize}
\end{theorem}
In view of the above theorem, we introduce the following definition.
\begin{definition}\label{DEF:yremainder}
A rational function $a/b\in\set F(y,z)$ with $a,b\in\set F[y,z]$ and $b\neq 0$ is called
a {\em $\sigma_y$-remainder} if $\deg_y(a)<\deg_y(b)$ and $b$ is $\sigma_y$-free.
\end{definition}
It is evident from
Theorem~\ref{THM:abred} that any nonzero $\sigma_y$-remainder is not
$\sigma_y$-summable.

Generalizing to the bivariate case, we consider the
$(\sigma_y,\sigma_z)$-summability problem of deciding whether a given
rational function $f\in \set F(y,z)$ can be written in the form $f =
\sigma_y(g)-g+\sigma_z(h)-h$ for $g,h\in \set F(y,z)$. If such a form
exists, we say that $f$ is {\em $(\sigma_y,\sigma_z)$-summable},
abbreviated as summable in certain instances. The
$(\sigma_y,\sigma_z)$-decomposition problem is then to decompose a
given rational function $f\in \set F(y,z)$ into the form $$f =
\sigma_y(g)-g+\sigma_z(h)-h+r,$$ where $g,h,r\in \set F(y,z)$ and $r$
is minimal in certain sense. Moreover, $f$ is
$(\sigma_y,\sigma_z)$-summable if and only if $r = 0$.

Recall \cite{AbLe2002} that an irreducible polynomial $f\in \set F[y,z]$ is called
{\em $(y,z)$-integer linear} over the field $\set F$ if it can be written in the
form $f = p(\alpha y+ \beta z)$ for a polynomial $p(Z)\in \set F[Z]$ and
integers $\alpha, \beta\in \set Z$. One may assume without loss of
generality that $\beta \geq 0$ and $\alpha,\beta$ are coprime. A
polynomial in $\set F[y,z]$ is called {\em $(y,z)$-integer linear}
over $\set F$ if all its irreducible factors are $(y,z)$-integer
linear over $\set F$ while a rational function in $\set F(y,z)$ is
called {\em $(y,z)$-integer linear} over $\set F$ if its numerator and
denominator are both $(y,z)$-integer linear over~$\set F$.  For
simplicity, we just say a rational function is $(y,z)$-integer linear
over $\set F$ of $(\alpha,\beta)$-type if it is equal to $p(\alpha y+
\beta z)$ for some $p(Z)\in \set F(Z)$ and $\alpha, \beta$ are coprime
integers with $\beta\geq 0$. Algorithms for determining integer
linearity can be found in \cite{AbLe2002,LiZh2013,GHLZ2019}.

In the context of creative telescoping, we will also need to consider
the variable $x$ and the automorphism $\sigma_x$, which
for every $f\in \set F(y,z)$ maps $f(x,y,z)$ to $f(x+1,y,z)$.
Two polynomials $p,q\in \set K[x,y,z]$ are called
{\em $(x,y,z)$-shift equivalent}, denoted by $p\sim_{x,y,z} q$, if there
exist three integers $\ell,m,n$ such that $p =
\sigma_x^\ell\sigma_y^m\sigma_z^n(q)$. We generalize this notion to the
domain $\set F[y,z]$ by saying that two polynomials $p,q\in \set F[y,z]$
are $(x,y,z)$-shift equivalent if $p=\sigma_x^\ell\sigma_y^m\sigma_z^n(q)$
for integers $\ell,m,n$. When $\ell =0$ then $p$ is also called
{\em $(y,z)$-shift equivalent} to $q$, denoted by $p \sim_{y,z} q$.
Clearly, both $\sim_{x,y,z}$ and $\sim_{y,z}$ are equivalence relations.

Let $\set F(y,z)[\shift_x,\shift_y,\shift_z]$ be the ring of linear
recurrence operators in $x,y,z$ over $\set F(y,z)$. Here
$\shift_x,\shift_y,\shift_z$ commute with each other, and $\shift_vf
= \sigma_v(f)\shift_v$ for any $f\in \set F(y,z)$ and $v\in\{x,y,z\}$.
The application of an operator
$P = \sum_{i,j,k\geq 0}p_{ijk}\,\shift_x^i\,\shift_y^j\,\shift_z^k$ in
$\set F(y,z)[\shift_x,\shift_y,\shift_z]$ to a rational function $f \in
\set F(y,z)$ is then defined as
\[P(f) = \sum_{i,j,k\geq 0}p_{ijk}\sigma_x^i\sigma_y^j\sigma_z^k(f).\]
\begin{definition}\label{DEF:telescoper}
  Let $f$ be a rational function in $\set F(y,z)$. A nonzero linear
  recurrence operator $L\in \set F[\emph\shift_x]$ is called a {\em telescoper}
  for $f$ if $L(f)$ is $(\sigma_y,\sigma_z)$-summable, or equivalently,
  if there exist rational functions $g,h\in \set F(y,z)$ such that
  \[
  L(f) = (\emph\shift_y-1)(g) + (\emph\shift_z-1)(h),
  \]
  where $1$ denotes the identity map of $\set F(y,z)$.  We call
  $(g,h)$ a corresponding certificate for $L$. The {\em order}
  of a telescoper is defined to be its degree in $\emph\shift_x$. A
  telescoper of minimal order for $f$ is called a
  {\em minimal telescoper} for $f$.
\end{definition}

\section{Bivariate extension of the Abramov reduction}\label{SEC:biabred}
In this section, we demonstrate how to solve the bivariate
decomposition problem (and thus also the bivariate summability
problem) using the Abramov reduction. To this end, let us first
recall some key results on the bivariate summability
from~\cite{HoWa2015}.

Based on a theoretical criterion given in \cite[Theorem~3.7]{ChSi2014}, Hou and
Wang~\cite{HoWa2015} developed an algorithm for solving the
$(\sigma_y,\sigma_z)$-summability problem.  The proof found in
\cite[Lemma 3.1]{HoWa2015} contains a reduction algorithm with inputs
and outputs specified as follows.

\smallskip\noindent {\bf Primary reduction.}
Given a rational function $f\in \set F(y,z)$, compute rational
functions $g, h, r \in \set F(y,z)$ such that
\begin{equation}\label{EQ:hwred}
  f = (\shift_y-1)(g) + (\shift_z-1)(h)  + r
\end{equation}
and $r$ is of the form
\begin{equation}\label{EQ:rform}
  r = \sum_{i=1}^m\sum_{j=1}^{n_i}\frac{a_{ij}}{b_{ij}d_i^j}
\end{equation}
with $m,n_i\in \set N$, $a_{ij}, d_i\in \set F[y,z]$ and $b_{ij}\in
\set F[y]$ satisfying that
\begin{itemize}
\item $\deg_z(a_{ij})<\deg_z(d_i)$,
\item $d_i$ is a monic irreducible factor of the denominator of $r$
  and of positive degree in $z$,
\item $d_i\nsim_{y,z} d_{i'}$ whenever $i\neq i'$ for $1\leq i,i'\leq m$.
\end{itemize}

Let $f$ be a rational function in $\set F(y,z)$ and assume that
applying the primary reduction to $f$ yields
\eqref{EQ:hwred}. Deciding if $f$ is $(\sigma_y,\sigma_z)$-summable
then amounts to checking the summability of $r$. By
\cite[Lemma~3.2]{HoWa2015}, this is equivalent to checking the
summability of each simple fraction $a_{ij}/(b_{ij}d_i^j)$. Thus the
bivariate summability problem for a general rational function is
reduced to determining the summability of several simple fractions, 
which in turn can be addressed by the following.
\begin{theorem}[{\cite[Theorem~3.3]{HoWa2015}}]\label{THM:hwthm3.3}
  Let $f = a/(b\, d^j)$, where $a, d\in \set F[y,z]$, $b\in \set F[y]$,
  $j\in \set N\setminus\{0\}$ with $d$ irreducible and $0\leq \deg_z(a)
  < \deg_z(d)$. Then $f$ is $(\sigma_y,\sigma_z)$-summable if and only if
  \begin{itemize}	
  \item[(i)] $d$ is $(y,z)$-integer linear over~$\set F$ of
    $(\alpha,\beta)$-type,

  \item[(ii)] there exists $q\in \set F(y)[z]$ with
    $\deg_z(q)<\deg_z(d)$ so that
    \begin{equation}\label{EQ:arelation}
      \frac{a}{b}=\sigma_y^\beta\sigma_z^{-\alpha}(q)-q.
    \end{equation}
  \end{itemize}
\end{theorem}
Since $d$ is irreducible, the first condition is easily recognized by
comparing coefficients once $d$ is known. In \cite[\S 4]{HoWa2015},
the second condition is checked by finding a polynomial solution of a
system of linear recurrence equations in one variable based on a
universal denominator derived from the $m$-fold Gosper
representation. Such a polynomial solution gives rise to a
desired $q$ for \eqref{EQ:arelation}.

In the rest of this section, we show how to detect the second
condition via the Abramov reduction, without solving any auxiliary
recurrence equations. As a result, we obtain an additive decomposition
of the given rational function in $\set F(y,z)$, from which one can
not only read off the $(\sigma_y,\sigma_z)$-summability, but also
gather useful descriptions on the possible \lq\lq minimal\rq\rq\
nonsummable part. This lays the foundation of our new algorithm
in Section~\ref{SEC:rct}.

Let $R$ be a ring (resp.\ field) and $\sigma: R\rightarrow R$ be an
automorphism of $R$. The pair $(R,\sigma)$ is called a
{\em difference ring (resp.\ field)}. An element $r\in R$ is called a
{\em constant} of $R$ with respect to $\sigma$ if $\sigma(r) = r$.
The set of all such constants forms a subring (resp.\ subfield) of
$R$, called the {\em constant subring (resp.\ subfield)} of $R$ with
respect to $\sigma$. Let $(R_1,\sigma_1)$ and $(R_2,\sigma_2)$ be two
difference rings.  A homomorphism (resp.\ isomorphism) $\psi: R_1
\rightarrow R_2$ is called a {\em difference homomorphism (resp.\ isomorphism)}
from $(R_1,\sigma_1)$ to $(R_2,\sigma_2)$ if $\sigma_2\circ\psi =
\psi\circ\sigma_1$, that is, the left diagram in Figure~\ref{FIG:rule}
commutes.  Two difference rings are then said to be {\em isomorphic}
if there exists a difference isomorphism between them.

\begin{figure}[t]
  \centering
  \begin{tikzpicture}
    \node at (0,0) {$R_2$};
    \draw[->] (0.8,0) -- (1.2,0) node[below] {$\sigma_2$} -- (1.6,0);
    \node at (2.4,0) {$R_2$};
    \draw[->] (0,1) -- (0,.75) node[left] {$\psi$} -- (0,.5);
    \node at (0,1.5) {$R_1$};
    \draw[->] (0.8,1.5) -- (1.2,1.5) node[above] {$\sigma_1$} -- (1.6,1.5);
    \node at (2.4,1.5) {$R_1$};
    \draw[->] (2.4,1) -- (2.4,.75) node[right] {$\psi$} -- (2.4,.5);
    \draw[->] (1,.9) arc [radius=.2, start angle = 140, end angle = -90];
    
    \node at (7,0) {$\set F(y,z)$};
    \draw[->] (7.8,0) -- (8.7,0) node[below] {$\sigma_y$} -- (9.6,0);
    \node at (10.4,0) {$\set F(y,z)$};
    \draw[->] (7,1) -- (7,.75) node[left] {$\phi_{\alpha,\beta}$} -- (7,.5);
    \node at (7,1.5) {$\set F(y,z)$};
    \draw[->] (7.8,1.5) -- (8.7,1.5) node[above]
         {$\tau=\sigma_y^\beta\sigma_z^{-\alpha}$} -- (9.6,1.5);
    \node at (10.4,1.5) {$\set F(y,z)$};
    \draw[->] (10.4,1) -- (10.4,.75) node[right] {$\phi_{\alpha,\beta}$} -- (10.4,.5);
    \draw[->] (8.6,.9) arc [radius=.2, start angle = 140, end angle = -90];
  \end{tikzpicture}
  \caption{Commutative diagrams for difference homomorphisms/isomorphisms.}
  \label{FIG:rule}
\end{figure}

Let $\alpha,\beta$ be two integers with $\beta$ nonzero. We define an
$\set F$-homomorphism $\phi_{\alpha,\beta}: \set F(y,z) \rightarrow \set F(y,z)$
by 
$$\phi_{\alpha,\beta}(y)=\beta y  ~~ \mbox{  and } ~~ \phi_{\alpha, \beta}(z)
= \beta^{-1}z -\alpha y.$$ It is readily seen that
$\phi_{\alpha,\beta}$ is an $\set F$-isomorphism with inverse
$\phi_{\alpha,\beta}^{-1}$ given by $$\phi_{\alpha,\beta}^{-1}(y) =
\beta^{-1}y ~~\mbox{  and } ~~ \phi_{\alpha,\beta}^{-1}(z) = \beta z +\alpha y.$$
We call $\phi_{\alpha,\beta}$ the {\em map for $(\alpha,\beta)$-shift reduction}.
\begin{proposition}\label{PROP:diffiso}
  Let $\alpha,\beta\in \set Z$ with $\beta \neq 0$ and $\tau =
  \sigma_y^\beta\sigma_z^{-\alpha}$. Then $\phi_{\alpha,\beta}$ is a
  difference isomorphism from $(\set F(y,z),\tau)$ to $(\set F(y,z),\sigma_y)$.
\end{proposition}
\begin{proof}
  Since $\phi_{\alpha,\beta}$ is an $\set F$-isomorphism, it remains
  to show that $\sigma_y\circ\phi_{\alpha,\beta}=\phi_{\alpha,\beta}\circ\tau$,
  namely the right diagram in Figure~\ref{FIG:rule} commutes.
  This is confirmed by the observation that
  \[\sigma_y(\phi_{\alpha,\beta}(f(y,z)))
  = \sigma_y(f(\beta y,\beta^{-1}z-\alpha y))
  = f(\beta y + \beta, \beta^{-1}z-\alpha y-\alpha)\]
  and
  \[\phi_{\alpha,\beta}(\tau(f(y,z)))
  = \phi_{\alpha,\beta}(f(y+\beta,z-\alpha))
  =  f(\beta y + \beta, \beta^{-1}z-\alpha y-\alpha)\]
  for any $f \in \set F(y,z)$.
\end{proof}

\begin{corollary}\label{COR:summableiso}
Let $f\in \set F(y,z)$ and assume the conditions of
Proposition~\ref{PROP:diffiso}.  Then $f$ is $\tau$-summable if and
only if $\phi_{\alpha,\beta}(f)$ is $\sigma_y$-summable.
\end{corollary}
\begin{proof}
  By Proposition~\ref{PROP:diffiso}, $\phi_{\alpha,\beta}$ is a
  difference isomorphism from $(\set F(y,z),\tau)$ to $(\set
  F(y,z),\sigma_y)$. It follows that
  \[f = \tau (g) - g
  \quad\Longleftrightarrow\quad
  \phi_{\alpha,\beta}(f) = \phi_{\alpha,\beta}(\tau(g)-g)
  = \sigma_y(\phi_{\alpha,\beta}(g))-\phi_{\alpha,\beta}(g)\]
  for any $g\in \set F(y,z)$. The assertion follows.
\end{proof}
The problem of deciding whether a rational function $f\in \set
F(y)[z]$ satisfies the equation \eqref{EQ:arelation}, that is, the
$\sigma_y^\beta\sigma_z^{-\alpha}$-summability problem for $f$, is
then equivalent to the $\sigma_y$-summability problem for
$\phi_{\alpha,\beta}(f)$. In fact, there is also a natural one-to-one
correspondence between additive decompositions of $f$ with respect to
$\sigma_y^\beta\sigma_z^{-\alpha}$ and additive decompositions of
$\phi_{\alpha,\beta}(f)$ with respect to $\sigma_y$. 
Together with Definition~\ref{DEF:yremainder}, this motivates
us to introduce the notions of remainder fractions and remainders, in
order to characterize nonsummable rational functions concretely.
\begin{definition}\label{DEF:remfrac}
  A fraction $a/(b\, d^j)$ with $a, d\in \set F[y,z]$, $b\in \set F[y]$
  and $j\in \set N\setminus\{0\}$ is called a {\em remainder fraction}
  if
  \begin{itemize}
  \item $\deg_z(a)<\deg_z(d)$;
  \item $d$ is monic, irreducible and of positive degree in $z$;
  \item $\phi_{\alpha,\beta}(a/b)$ is a $\sigma_y$-remainder
    in case $d$ is $(y,z)$-integer linear over $\set F$ of
    $(\alpha,\beta)$-type.
  \end{itemize}
\end{definition}

\begin{definition}\label{DEF:remainder}
  Let $f$ be a rational function in $\set F(y,z)$. Then $r\in \set
  F(y,z)$ is called a {\em $(\sigma_y,\sigma_z)$-remainder} of $f$ if
  $f-r$ is $(\sigma_y,\sigma_z)$-summable
  and $r$ can be written as
  \begin{equation}\label{EQ:remainder}
    r=\sum_{i=1}^m\sum_{j=1}^{n_i}\frac{a_{ij}}{b_{ij}d_i^j},
  \end{equation}
  where $m,n_i\in\set N$, $a_{ij},d_i\in\set F[y,z]$, $b_{ij}\in \set
  F[y]$ with each $a_{ij}/(b_{ij}d_i^j)$ being a remainder fraction,
  $d_i$ being a factor of the denominator of $r$, and
  $d_i\nsim_{y,z}d_{i'}$ whenever $i\neq i'$ and $1\leq i,i'\leq m$.
  For brevity, we just say that $r$ is a
  $(\sigma_y,\sigma_z)$-remainder if $f$ is clear from the
  context. We refer to~\eqref{EQ:remainder}, along
  with the attached conditions, as the {\em remainder form} of $r$.
\end{definition}
The uniqueness of partial fraction decompositions (in this case with
respect to~$z$) implies that the remainder form of a given
$(\sigma_y,\sigma_z)$-remainder is unique up to reordering and
multiplication by units of $\set F$. Evidently, every single remainder
fraction, or part of summands in \eqref{EQ:remainder}, is a
$(\sigma_y,\sigma_z)$-remainder.  Remainders not only help us to
recognize summability, but also describe the \lq\lq
minimum\rq\rq\ gap between a given rational function and summable
rational functions, as shown in the next two propositions.
\begin{proposition}\label{PROP:remainder}
  Let $r\in\set F(y,z)$ be a nonzero $(\sigma_y,\sigma_z)$-remainder
  with the form~\eqref{EQ:remainder}.  Then each nonzero
  $a_{ij}/(b_{ij}d_i^j)$ for $1\leq i\leq m$ and $1\leq j\leq n_i$, as
  well as $r$ itself, is not $(\sigma_y,\sigma_z)$-summable.
\end{proposition}
\begin{proof}
  Since $r$ is a $(\sigma_y,\sigma_z)$-remainder, each
  $a_{ij}/(b_{ij}d_i^j)$ is a remainder fraction. For a particular
  nonzero $a_{ij}/(b_{ij}d_i^j)$, namely $a_{ij}\neq 0$, we claim that
  it is not $(\sigma_y,\sigma_z)$-summable.  If $d_i$ is not
  $(y,z)$-integer linear over~$\set F$, then the claim follows by Theorem~\ref{THM:hwthm3.3}.
  Otherwise, assume that $d_i$ is $(y,z)$-integer linear over $\set F$
  of $(\alpha,\beta)$-type.  Since $a_{ij}/(b_{ij}d_i^j)$ is a
  remainder fraction, Definition~\ref{DEF:remfrac} reads that
  $\phi_{\alpha,\beta}(a_{ij}/b_{ij})$ is a $\sigma_y$-remainder and
  thus is not $\sigma_y$-summable. By Corollary~\ref{COR:summableiso},
  $a_{ij}/b_{ij}$ is not $\sigma_y^\beta\sigma_z^{-\alpha}$-summable.
  The claim is then again assured by Theorem~\ref{THM:hwthm3.3}.

  In either case, we have that $a_{ij}/(b_{ij}d_i^j)$ is not
  $(\sigma_y,\sigma_z)$-summable.  Since $r$ is nonzero, at least one
  of the $a_{ij}/(b_{ij}d_i^j)$ is nonzero. By \cite[Lemma~3.2]{HoWa2015},
  $r$ is therefore not $(\sigma_y,\sigma_z)$-summable.
\end{proof}
\begin{proposition}\label{PROP:minimality}
  Let $r\in \set F(y,z)$ be a nonzero $(\sigma_y,\sigma_z)$-remainder
  with the form~\eqref{EQ:remainder}, in which $a_{ij}$ and
  $b_{ij}d_i^j$ are further assumed to be coprime. Assume that there
  exists another $r' \in \set F(y,z)$ such that $r'-r$ is
  $(\sigma_y,\sigma_z)$-summable. Write $r'$ in the form
  \[r' = p' + \sum_{i=1}^{m'}\sum_{j=1}^{n'_i}
  \frac{a'_{ij}}{b'_{ij}{d'_i}^j},\]
  where $m',n'_i\in \set N$, $p'\in \set F(y)[z]$, $a'_{ij},d'_i\in
  \set F[y,z]$ and $b'_{ij}\in \set F[y]$ with
  $\deg_z(a'_{ij})<\deg_z(d'_i)$ and $d'_i$ being monic irreducible
  factor of the denominator of $r'$ and of positive degree in $z$.
  For each $1\leq i\leq m$, define
  \[\Lambda_i = \{i'\in \set N\mid 1\leq i'\leq m'
  \ \text{and}\ d'_{i'}=\sigma_y^{\lambda_{i'}}
  \sigma_z^{\mu_{i'}}(d_i)\ \text{for }\lambda_{i'},\mu_{i'}
  \in \set Z\}.\]
  Then $\Lambda_i$ is nonempty for any $1\leq i\leq m$. Moreover,
  $m\leq m'$, $n_i\leq n'_{i'}$ for all $i' \in
  \Lambda_i$, $\deg_y(b_{ij})\leq\sum_{i'\in\Lambda_i}
  \deg_y(b'_{i'j})$ for each $1\leq i\leq m$ and $1\leq
  j\leq n_i$, and the degree in $z$ of the denominator of $r$ is no
  more than that of $r'$.
\end{proposition}
\begin{proof}
  Since $r'-r$ is $(\sigma_y,\sigma_z)$-summable, all the rational
  function $\sum_{i'\in \Lambda_{i}} a'_{i'j}/(b'_{i'j}{d'_{i'}}^j)
  -a_{ij}/(b_{ij}d_i^j)$ are $(\sigma_y,\sigma_z)$-summable by
  \cite[Lemma~3.2]{HoWa2015}, and then so are the
  \begin{equation}\label{EQ:summablerat}
    {\sum_{i'\in \Lambda_{i}}
      \frac{\sigma_y^{-\lambda_{i'}}\sigma_z^{-\mu_{i'}}(a'_{i'j})}
           {\sigma_y^{-\lambda_{i'}}(b'_{i'j}){d_i^j}}
      -\frac{a_{ij}}{b_{ij} d_{i}^j}}.
  \end{equation}
  Since $r$ is a nonzero $(\sigma_y,\sigma_z)$-remainder, we conclude
  from Proposition~\ref{PROP:remainder} that each nonzero
  $a_{ij}/(b_{ij}d_i^j)$ is not $(\sigma_y,\sigma_z)$-summable.
  Notice that for each $1\leq i\leq m$, there is at least one integer
  $j$ with $1\leq j\leq n_i$ such that $a_{ij}\neq 0$. It then follows
  from the summability of \eqref{EQ:summablerat} that every
  $\Lambda_i$ is nonempty, namely every $d_i$ is $(y,z)$-shift
  equivalent to some $d'_{i'}$ for $1\leq i'\leq m'$, and that
  $n_i\leq n'_{i'}$ for any $i'\in \Lambda_i$. Notice that the $d_i$
  are pairwise $(y,z)$-shift inequivalent. Thus the $\Lambda_i$ are
  pairwise disjoint, which implies that $m\leq m'$. Accordingly, the
  degree in $z$ of the denominator of $r$ is no more than that
  of~$r'$.

  It remains to show the inequality for the degree of each
  $b_{ij}$. For each $1\leq i\leq m$ and $1\leq j\leq n_i$, by
  Theorem~\ref{THM:hwthm3.3}, the summability of
  \eqref{EQ:summablerat} either yields
  \[\sum_{i'\in \Lambda_i}\frac{\sigma_y^{-\lambda_{i'}}
    \sigma_z^{-\mu_{i'}}(a'_{i'j})}
    {\sigma_y^{-\lambda_{i'}}(b'_{i'j})}
    = \sigma_y^\beta\sigma_z^{-\alpha}(q)-q+\frac{a_{ij}}{b_{ij}}
    \quad \text{for some}\ q\in \set F(y)[z],\]
  if $d_i$ is $(y,z)$-integer linear over $\set F$ of
  $(\alpha,\beta)$-type or otherwise yields
  \[\sum_{i'\in \Lambda_i}
  \frac{\sigma_y^{-\lambda_{i'}}\sigma_z^{-\mu_{i'}}(a'_{\ i'j})}
       {\sigma_y^{-\lambda_{i'}}(b'_{i'j})}
       =\frac{a_{ij}}{b_{ij}}.\]
  The assertion is evident in the latter case. For the former case,
  because $a_{ij}/(b_{ij}d_i^j)$ is a remainder fraction, the
  assertion follows by the minimality of $\phi_{\alpha,\beta}(b_{ij})$
  (and thus $b_{ij}$) from Theorem~\ref{THM:abred}.
\end{proof}

With everything in place, we now present a bivariate extension of the
Abramov reduction, which addresses the
$(\sigma_y,\sigma_z)$-decomposition problem.

\medskip\noindent{\bf Bivariate Abramov reduction.}
Given a rational function $f\in \set F(y,z)$, compute three rational
functions $g,h,r\in \set F(y,z)$ such that $r$ is a
$(\sigma_y,\sigma_z)$-remainder of $f$ and
\begin{equation}\label{EQ:biabred}
  f = (\shift_y-1)(g)+(\shift_z-1)(h)+r.
\end{equation}

\begin{enumerate}
\item Apply the primary reduction to $f$ to find $g,h\in \set F(y,z)$,
  $m,n_i\in \set N$, $a_{ij},d_i\in \set F[y,z]$ and $b_{ij}\in \set
  F[y]$ such that \eqref{EQ:hwred} holds.

\item For $i = 1, \dots, m$ do\\[1ex]
  \hphantom{for}If $d_i$ is $(y,z)$-integer linear over $\set F$
  of $(\alpha_i,\beta_i)$-type then
  \begin{itemize}
  \item[2.1] Compute $\tilde a_{ij}/\tilde b_{ij} =
    \phi_{\alpha_i,\beta_j}(a_{ij}/b_{ij})$ with
    $\phi_{\alpha_i,\beta_i}$ being the map for
    $(\alpha_i,\beta_i)$-shift reduction;
    \item[2.2] For $j=1, \dots, n_i$ do
    \begin{itemize}
    \item[2.2.1] Apply the Abramov reduction to $\tilde
      a_{ij}/\tilde b_{ij}$ with respect to $y$ to get $\tilde
      q_{ij},\tilde r_{ij}\in \set F(y)[z]$ such that
      \begin{equation*}
        \frac{\tilde a_{ij}}{\tilde b_{ij}} =
        \sigma_y(\tilde q_{ij})-\tilde q_{ij}+\tilde r_{ij}.
      \end{equation*}
    \item[2.2.2] Apply $\phi_{\alpha_i,\beta_i}^{-1}$ to both sides of the previous
      equation to get
      \begin{equation}\label{EQ:bishiftred}
        \frac{a_{ij}}{b_{ij}}
        =\sigma_y^{\beta_i}\sigma_z^{-\alpha_i}(q_{ij})-q_{ij}+r_{ij},
      \end{equation}
      where $q_{ij}=\phi_{\alpha_i,\beta_i}^{-1}(\tilde q_{ij})$ and
      $r_{ij} = \phi_{\alpha_i,\beta_i}^{-1}(\tilde r_{ij})$.
    \item[2.2.3] Update $a_{ij}$ and $b_{ij}$ to be the numerator and
      denominator of $r_{ij}$, respectively.
    \end{itemize}
  \item[2.3] Update
    \begin{equation}\label{EQ:gh}
    g = g + \sum_{j=1}^{n_i}\sum_{k=0}^{\beta_i-1}
      \sigma_y^k\sigma_z^{-\alpha_i}\left(\frac{q_{ij}}{d_i^j}\right)
      \ \ \text{and}\ \
      h = h+
      \begin{cases}
        {\displaystyle
        \sum_{j=1}^{n_i}\sum_{k=1}^{\alpha_i}\sigma_z^{-k}\left(
        -\frac{q_{ij}}{d_i^j}\right)} \quad &\alpha_i\geq 0\\[4ex]
        {\displaystyle
        \sum_{j=1}^{n_i}\sum_{k=0}^{-\alpha_i-1}\sigma_z^k\left(
        \frac{q_{ij}}{d_i^j}\right)}\quad &\alpha_i< 0
      \end{cases}.
      \end{equation}
  \end{itemize}
\item Update $r = \sum_{i=1}^m\sum_{j=1}^{n_i}a_{ij}/(b_{ij}d_i^j)$, and
  return $g,h,r$.
\end{enumerate}

\begin{theorem}\label{THM:biabred}
  Let $f$ be a rational function in $\set F(y,z)$. Then the bivariate
  Abramov reduction computes two rational functions $g,h\in \set
  F(y,z)$ and a $(\sigma_y,\sigma_z)$-remainder $r\in \set F(y,z)$
  such that~\eqref{EQ:biabred} holds. Moreover, $f$ is
  $(\sigma_y,\sigma_z)$-summable if and only if $r = 0$.
\end{theorem}
\begin{proof}
  Applying the primary reduction to $f$ yields \eqref{EQ:hwred}.  For
  any nonzero $a_{ij}/(b_{ij}d_i^j)$ obtained in step~1, if $d_i$ is
  not $(y,z)$-integer linear over $\set F$ then we know from
  Theorem~\ref{THM:hwthm3.3} that $a_{ij}/(b_{ij}d_i^j)$ is not
  $(\sigma_y,\sigma_z)$-summable and is thus already a remainder
  fraction. Otherwise, there exist coprime integers $\alpha_i,\beta_i$
  with $\beta_i>0$ so that $d_i= p_i(\alpha_i y + \beta_i z)$ for some
  $p_i(Z)\in \set F[Z]$. By Theorem~\ref{THM:abred} and Definition~\ref{DEF:remfrac}, for
  each $1\leq j\leq n_i$, steps~2.2.1-2.2.2 correctly find $q_{ij}$
  and $r_{ij}$ such that \eqref{EQ:bishiftred} holds and
  $r_{ij}/d_i^j$ is a remainder fraction. After step~2.2, plugging all
  \eqref{EQ:bishiftred} into \eqref{EQ:hwred} then gives (with a
  slight abuse of notation):
  \[f = (\shift_y-1)(g)+(\shift_z-1)(h) +
  \sum_{i:\, d_i = p_i(\alpha_i y+ \beta_i z)}\sum_{j=1}^{n_i}
  \frac{\sigma_y^{\beta_i}\sigma_z^{-\alpha_i}(q_{ij})-q_{ij}}{d_i^j}
  + r,\]
  where the index $i$ runs through all $(y,z)$-integer linear $d_i$'s
  and $r = \sum_{i=1}^m\sum_{j=1}^{n_i}a_{ij}/(b_{ij}d_i^j)$ is a
  $(\sigma_y,\sigma_z)$-remainder by Definition~\ref{DEF:remainder}.
  The assertions then follow from Proposition~\ref{PROP:remainder} and
  the observation that
  \begin{align}
    \frac{\sigma_y^{\beta_i}\sigma_z^{-\alpha_i}(q_{ij})-q_{ij}}{d_i^j}
    =(\shift_y-1)\left(\sum_{k=0}^{\beta_i-1}\sigma_y^k\sigma_z^{-\alpha_i}
     \left(\frac{q_{ij}}{d_i^j}\right)\right)+
   \begin{cases}
     {\displaystyle
     (\shift_z-1)\left(\sum_{k=1}^{\alpha_i}\sigma_z^{-k}
     \left(-\frac{q_{ij}}{d_i^j}\right)\right)}
     &\text{if } \alpha_i\geq 0 \\[4ex]
     {\displaystyle
    (\shift_z-1)\left(\sum_{k=0}^{-\alpha_i-1}\sigma_z^{k}
     \left(\frac{q_{ij}}{d_i^j}\right)\right)}
     &\text{if }\alpha_i< 0
   \end{cases}
   \label{EQ:integerlinearsummable}
  \end{align}
  for any $d_i = p_i(\alpha_i y+\beta_i z)$.
\end{proof}

\begin{example}\label{EX:example}
  Consider the rational function $f$ admitting the partial fraction
  decomposition $f = f_1 + f_2 + f_3$ with
  \begin{align*}
    f_1 &= \frac{x^2z+1}{\underbrace{(x+y)(x+z)^2+1}_{d_1}},\ \
    f_2 = \frac{(x^2+xy+3x-3)z-x-y+3}{(x+y)(x+y+3)
      (\underbrace{(x+2y+3z)^2+1}_{d_2})}
    \ \ \text{and}\ \
    f_3 = \frac1{\underbrace{x-y+z}_{d_3}}.
  \end{align*}
  Note that $d_1,d_2,d_3$ are $(y,z)$-shift inequivalent to each
  other. Hence $f$ remains unchanged after applying the primary
  reduction. Since $d_1$ is not $(y,z)$-integer linear, we leave $f_1$
  untouched and proceed to deal with $f_2$. Notice that $d_2$ is
  $(y,z)$-integer linear of $(2,3)$-type. Then applying the Abramov
  reduction to $\phi_{2,3}(f_2 d_2)$ with $\phi_{2,3}$ being the map 
  for $(2,3)$-shift reduction yields
  \[ 
  \phi_{2,3}(f_2 d_2)
  = (\shift_y-1)\left(\frac{z-6xy^2-2x^2y+2x}{3(x+3y)}\right)+
  \frac{\frac13 xz+\frac23 x^2+1}{x+3y},\]
  which, when applied by $\phi_{2,3}^{-1}$, leads to
  \[
  f_2 d_2= (\shift_y^3\shift_z^{-2}-1)(q_2)
  +\frac{\frac13x(2y+3z)+\frac23x^2+1}{x+y}
  \quad\text{with}\ q_2=\frac{3(2y+3z)-2xy^2-2x^2y+6x}{9(x+y)} .\]  
  Using \eqref{EQ:integerlinearsummable}, we decompose $f_2$ as
  \[ f_2 = (\shift_y-1)\left(\sum_{k=0}^{2}\sigma_y^k\sigma_z^{-2}
  \Big(\frac{q_2}{d_2}\Big)\right) + (\shift_z-1) \left(
  \sum_{k=1}^2\sigma_z^{-k}\Big(-\frac{q_2}{d_2}\Big)\right)+
  \underbrace{\frac{\frac13x(2y+3z)+\frac23x^2+1}{(x+y)((x+2y+3z)^2+1)}}_{r}.\]
  One sees that $r$ is a $(\sigma_y,\sigma_z)$-remainder of $f_2$, and
  thus $f_2$ is not $(\sigma_y,\sigma_z)$-summable by
  Theorem~\ref{THM:biabred}.  Along the same lines as above, we have
  \[f_3 = (\shift_y-1)\left(\frac{y}{x-y+z+1}\right)
  +(\shift_z-1)\left(\frac{y}{x-y+z}\right), \]
  implying that $f_3$ is $(\sigma_y,\sigma_z)$-summable. Combining
  everything together, $f$ is finally decomposed as
  \[f = (\shift_y-1)(g)+(\shift_z-1)(h) +f_1 +r\]  
  with $g=\sum_{k=0}^{2}\sigma_y^k\sigma_z^{-2}(q_2/d_2) +y/(x-y+z+1)$
  and $h=\sum_{k=1}^2\sigma_z^{-k}(-q_2/d_2) +y/(x-y+z)$. Thus $f$ is
  not $(\sigma_y,\sigma_z)$-summable by Theorem~\ref{THM:biabred}. We
  will use $f$ as a running example in this paper.
\end{example}

\section{Linearity of remainders}\label{SEC:remlinearity}
As mentioned in the introduction, we expect our reduction algorithm to
induce a linear map, that is, the sum of two remainders was expected
to also be a remainder. Unfortunately, this is not always the case in
our setting, because some requirements in Definition~\ref{DEF:remainder}
may be broken by the addition among $(\sigma_y,\sigma_z)$-remainders, as seen in the
following examples. This prevents us from applying the bivariate
Abramov reduction developed in the previous section to construct a
telescoper in a direct way as was done in the differential case.
However, observe that a rational function in $\set F(y,z)$ may have
more than one $(\sigma_y,\sigma_z)$-remainder and any two of them
differ by a $(\sigma_y,\sigma_z)$-summable rational function. This
suggests a possible way to circumvent the above difficulty. That is,
choosing proper members from the residue class modulo summable
rational functions of the given $(\sigma_y,\sigma_z)$-remainders so as to make the linearity
become true. The goal of this section is to show that this direction
always works and it can be accomplished algorithmically.  We note that
a similar idea was used in the bivariate hypergeometric case \cite[\S 5]{CHKL2015}.
\begin{example}\label{EX:nonilremsum}
  Let $r = f_1$ and $s = \sigma_x(f_1)$ with $f_1$ being given in
  Example~\ref{EX:example}. Then $r$ and $s$ are both
  $(\sigma_y,\sigma_z)$-remainders since both denominators $d_1$ and
  $\sigma_x(d_1)$ are not $(y,z)$-integer linear. However their sum is
  not a $(\sigma_y,\sigma_z)$-remainder since $d_1$ is $(y,z)$-shift
  equivalent to $\sigma_x(d_1)$, namely $d_1 =
  \sigma_y^{-1}\sigma_z^{-1}\sigma_x(d_1)$.  Nevertheless, we can find
  another $(\sigma_y,\sigma_z)$-remainder $t$ of $s$ such that $r+t$
  has this property. For example, let
  \[t =  (\shift_y-1)\left(-\sigma_y^{-1}(s)\right)
  +(\shift_z-1)\left(-\sigma_y^{-1}\sigma_z^{-1}(s)\right) +s
  =\frac{(x+1)^2(z-1)+1}{(x+y)(x+z)^2+1},\]
  and then
  \[r+t = \frac{(2x^2+2x+1)z-x^2-2x+1}{(x+y)(x+z)^2+1}\]
  is a $(\sigma_y,\sigma_z)$-remainder by definition.  Alternatively,
  one can compute a $(\sigma_y,\sigma_z)$-remainder $\tilde t$ of $r$,
  say
  \begin{align*}
    \tilde t = (\shift_y-1)\left(r\right)
    +(\shift_z-1)\left(\sigma_y(r)\right)
    + r = \frac{x^2(z+1)+1}{(x+y+1)(x+z+1)^2+1}
  \end{align*}
  so that
  \[\tilde t+ s= \frac{(2x^2+2x+1)z+x^2+2}{(x+y+1)(x+z+1)^2+1}\]
  is a $(\sigma_y,\sigma_z)$-remainder.
\end{example}
\begin{example}\label{EX:ilremsum}
  Let
  \[r =\frac{\frac13x(2y+3z)+\frac23x^2+1}{(x+y)((x+2y+3z)^2+1)}
  \quad\text{and}\quad s = \frac{(\frac13x+1)(2y+3z)
    +\frac23(x+1)^2+2x+\frac{13}3}{(x+y+5)((x+2y+3z+1)^2+1)}.\]
  Then both $r$ and $s$ are $(\sigma_y,\sigma_z)$-remainders, but
  again their sum is not since $(x+2y+3z)^2+1$ is $(y,z)$-shift
  equivalent to $(x+2y+3z+1)^2+1$.  As in Example~\ref{EX:nonilremsum},
  we find a $(\sigma_y,\sigma_z)$-remainder
  \[\tilde s
  =\frac{a/b}{(x+2y+3z)^2+1}\quad\text{with} \quad \frac{a}{b} =
  \frac{(\frac13x+1)(2y+3z)+\frac23x^2+3x+4}{x+y+6}\]
  such that $s-\tilde s$ is
  $(\sigma_y,\sigma_z)$-summable.  However, the sum $r + \tilde s$ is
  still not a $(\sigma_y,\sigma_z)$-remainder since
  $\phi_{2,3}\big((\frac13x(2y+3z)+\frac23x^2+1)/(x+y)+a/b\big)$ is
  not a $\sigma_y$-remainder, where $\phi_{2,3}$ denotes the map for
  $(2,3)$-shift reduction. Notice that
  \begin{align*}
    \frac{a}{b}
    = (\shift_y^3\shift_z^{-2}-1)
    \left(\sum_{k=1}^2\sigma_y^{-3k}\sigma_z^{2k}
    \Big(\frac{a}{b} \Big)\right)
    +\frac{(\frac13x+1)(2y+3z)+\frac23x^2+3x+4}{x+y},
  \end{align*}
  so \eqref{EQ:integerlinearsummable} enables us to find a new
  $(\sigma_y,\sigma_z)$-remainder
  \[t=\frac{(\frac13x+1)(2y+3z)+\frac23x^2+3x+4}
           {(x+y)((x+2y+3z)^2+1)}\]
  such that $s-t$ is $(\sigma_y,\sigma_z)$-summable and
  \[r+t = \frac{(\frac23x+1)(2y+3z)+\frac43x^2+3x+5}{(x+y)((x+2y+3z)^2+1)}\]
  is a $(\sigma_y,\sigma_z)$-remainder. Another possible choice is to
  find a $(\sigma_y,\sigma_z)$-remainder $\tilde r$ of~$r$ such that
  $\tilde r + s$ is a $(\sigma_y,\sigma_z)$-remainder.
\end{example}
In order to achieve the linearity of $(\sigma_y,\sigma_z)$-remainders, we need to develop
two lemmas. The first one mimics the idea of Lemma~5.5
in~\cite{CHKL2015} in the bivariate setting.

\begin{lemma}\label{LEM:localremainder}
  Let $a,d\in \set F[y,z], b\in \set F[y]\setminus\{0\}$ and $j\in
  \set N\setminus\{0\}$.  Let $\lambda,\mu$ be two integers. Then one
  finds $g,h\in \set F(y,z)$ such that
  \begin{equation}\label{EQ:remsum}
    \frac{a}{b\, d^j} = (\shift_y-1)(g)+(\shift_z-1)(h)
    +\frac{\sigma_y^\lambda\sigma_z^\mu(a)}
    {\sigma_y^\lambda(b)\sigma_y^\lambda\sigma_z^\mu(d)^j}.
  \end{equation}
  Moreover, assume that $d$ is not $(y,z)$-integer linear over $\set F$. If
  $a/(b\, d^j)$ is a remainder fraction, then so is
  $\sigma_y^\lambda\sigma_z^\mu(a)/(\sigma_y^\lambda(b)\sigma_y^\lambda\sigma_z^\mu(d)^j)$.
\end{lemma}
\begin{proof}
  A direct calculation shows that
  \begin{align*}
    \frac{s}{t} &= (\shift_v-1)\left(
    -\sum_{j=0}^{i-1}\sigma_v^j\left(\frac{s}{t}\right)\right)
    +\frac{\sigma_v^i(s)}{\sigma_v^i(t)}
    =(\shift_v-1)\left(
    \sum_{j=1}^i\sigma_v^{-j}\left(\frac{s}{t}\right)\right)
    +\frac{\sigma_v^{-i}(s)}{\sigma_v^{-i}(t)}
  \end{align*}
  for any $s,t\in \set F[y,z]$, $i\in \set N$ and $v\in \{y,z\}$. By
  iteratively applying the above formulas, one readily computes
  $g,h\in\set F(y,z)$ such that \eqref{EQ:remsum} holds.
  
  Moreover, if $d$ is not $(y,z)$-integer linear over $\set F$, then
  neither is $\sigma_y^\lambda\sigma_z^\mu(d)$.  Since $a/(b\,d^j)$ is a
  remainder fraction, by Definition~\ref{DEF:remfrac},
  $\deg_z(a)<\deg_z(d)$ and $d$ is monic, irreducible and of positive
  degree in~$z$.  Shifting polynomials in $\set F[y,z]$ with respect
  to $y$ or $z$ preserves these properties.  It follows from
  definition that
  $\sigma_y^\lambda\sigma_z^\mu(a)/(\sigma_y^\lambda(b)\sigma_y^\lambda\sigma_z^\mu(d)^j)$
  is a remainder fraction.
\end{proof}
The next lemma is an immediate result of Theorem~5.6 in
\cite{CHKL2015}.
\begin{lemma}\label{LEM:bishiftcoprime}
  Let $\alpha,\beta\in \set Z$ with $\beta \neq 0$ and let
  $\phi_{\alpha,\beta}$ denote the map for $(\alpha,\beta)$-shift
  reduction. Let $a,\bar a\in \set F[y,z]$ and $b,\bar b \in \set
  F[y]\setminus\{0\}$ be such that both $\phi_{\alpha,\beta}(a/b)$ and
  $\phi_{\alpha,\beta}(\bar a/\bar b)$ are $\sigma_y$-remainders. Then
  one finds $q\in\set F(y)[z]$, $a'\in \set F[y,z]$ and $b'\in \set
  F[y]$ with $\phi_{\alpha,\beta}(a'/b')$ being a $\sigma_y$-remainder
  such that
  \[
  \frac{a}{b} = (\shift_y^\beta\shift_z^{-\alpha}-1)(q)
  + \frac{a'}{b'},
  \]
  and $\phi_{\alpha,\beta}(c_1\bar a/\bar b+c_2a'/b')$ is a
  $\sigma_y$-remainder for all $c_1,c_2\in\set F$.
\end{lemma}
\begin{proof}
 By \cite[Theorem~5.6]{CHKL2015} and \cite[Proposition~3.2]{Huan2016},
 there exist $\tilde q\in \set F(y)[z]$, $\tilde a\in \set F[y,z]$ and
 $\tilde b\in \set F[y]$ with $\tilde a/\tilde b$ being a
 $\sigma_y$-remainder such that
 \[
 \phi_{\alpha,\beta}\left(\frac{a}{b}\right)
 = \sigma_y(\tilde q)-\tilde q + \frac{\tilde a}{\tilde b},
 \]
 and $c_1\phi_{\alpha,\beta}(\bar a/\bar b)+c_2\tilde a/\tilde b$ is
 a $\sigma_y$-remainder for all $c_1,c_2\in\set F$.  Notice that
 $\phi_{\alpha,\beta}$ is an $\set F$-isomorphism and
 $\sigma_y\circ\phi_{\alpha,\beta} = \phi_{\alpha,\beta}\circ \tau$
 with $\tau = \sigma_y^\beta\sigma_z^{-\alpha}$.  So
 $\phi_{\alpha,\beta}^{-1}\circ\sigma_y =
 \tau\circ\phi_{\alpha,\beta}^{-1}$.  Letting
 $q=\phi_{\alpha,\beta}^{-1}(\tilde q)$,
 $a'=\phi_{\alpha,\beta}^{-1}(\tilde a)$ and
 $b'=\phi_{\alpha,\beta}^{-1}(\tilde b)$ concludes the lemma.
\end{proof}

We are now ready to give an algorithm that provides a feasible way to
obtain the linearity.

\medskip\noindent {\bf Remainder linearization.}
Given two $(\sigma_y,\sigma_z)$-remainders $r, s\in \set F(y,z)$,
compute $g, h\in \set F(y,z)$ and a
$(\sigma_y,\sigma_z)$-remainder $t\in \set F(y,z)$ such that
 \begin{equation}\label{EQ:s2t}
    s=(\shift_y-1)(g)+(\shift_z-1)(h)+t
\end{equation}
and $c_1r+c_2t$ is a $(\sigma_y,\sigma_z)$-remainder for all
$c_1,c_2\in\set F$.

\begin{enumerate}
\item Write $r$ and $s$ in the remainder forms
  \begin{equation}\label{EQ:rsforms}
    r = \sum_{i=1}^{\bar m}\sum_{j=1}^{\bar n_i}
    \frac{\bar a_{ij}}{\bar b_{ij}\bar d_i^j}
    \quad\text{and}\quad
    s = \sum_{i=1}^{m}\sum_{j=1}^{n_i} \frac{a_{ij}}{b_{ij}d_i^j}.
  \end{equation}
	
\item Set $g = h= 0$.\\[1ex]
  For $i = 1, \dots, m$ do\\[1ex]
  \hphantom{for}If there exists $k\in \{1, 2, \dots, \bar m\}$ such
  that $\bar d_{k}=\sigma_y^\lambda\sigma_z^\mu(d_i)$ for some
  $\lambda, \mu\in \set Z$, then
  \begin{itemize}
  \item[2.1] For $j=1,\dots, n_i$ do
  \begin{itemize}
  \item[2.1.1] Apply Lemma~\ref{LEM:localremainder} to 
  $a_{ij}/(b_{ij}d_i^j)$ to find $g_{ij}, h_{ij}\in \set F(y,z)$ such that
  \begin{equation}\label{EQ:simplefrac}
    \frac{a_{ij}}{b_{ij}d_i^j}=(\shift_y-1)(g_{ij})+(\shift_z-1)(h_{ij})
    +\frac{\sigma_y^\lambda\sigma_z^\mu(a_{ij})}{\sigma_y^\lambda(b_{ij})\bar d_k^j}.
  \end{equation}
  
  \item[2.1.2] If $d_i$ is $(y,z)$-integer linear over $\set F$ of
    $(\alpha_i,\beta_i)$-type then
  \begin{quote}
  Apply Lemma~\ref{LEM:bishiftcoprime} to
  $\sigma_y^\lambda\sigma_z^\mu(a_{ij})/\sigma_y^\lambda(b_{ij})$ to
  find $q_{ij}\in \set F(y)[z]$, $a_{ij}'\in \set F[y,z]$ and $b_{ij}'
  \in \set F[y]$ with $\phi_{\alpha_i,\beta_i}(a_{ij}'/b_{ij}')$ being
  a $\sigma_y$-remainder such that
  \begin{equation}\label{EQ:bishift}
    \frac{\sigma_y^\lambda\sigma_z^\mu(a_{ij})}{\sigma_y^\lambda(b_{ij})}
    =(\shift_y^{\beta_i}\shift_z^{-\alpha_i}-1)(q_{ij})+\frac{a_{ij}'}{b_{ij}'},
  \end{equation}
  and $\phi_{\alpha_i,\beta_i}(c_1\bar a_{kj}/\bar b_{kj}+c_2a_{ij}'/b_{ij}')$
  is a $\sigma_y$-remainder for all $c_1,c_2\in\set F$; update
  $a_{ij}, b_{ij}$ to be $a'_{ij}, b'_{ij}$, respectively, and update
  $g,h$ by \eqref{EQ:gh}.
  \end{quote}
  
  Else update $a_{ij},b_{ij}$ to be
  $\sigma_y^\lambda\sigma_z^\mu(a_{ij}), \sigma_y^\lambda(b_{ij})$,
  respectively.
  \end{itemize}
  \item[2.2] Update $d_i$ to be $\bar d_k$, and update $g =
    g+\sum_{j=1}^{n_{i}}g_{ij}$, $h= h+\sum_{j=1}^{n_i}h_{ij}$.
  \end{itemize}
\item Set $t = \sum_{i=1}^{m}\sum_{j=1}^{n_i}a_{ij}/(b_{ij}d_i^j)$,
  and return $g,h,t$.
\end{enumerate}

\begin{theorem}\label{THM:remsum}
  Let $r$ and $s$ be two $(\sigma_y,\sigma_z)$-remainders. Then the
  remainder linearization correctly finds two rational functions
  $g,h\in \set F(y,z)$ and a $(\sigma_y,\sigma_z)$-remainder $t\in
  \set F(y,z)$ such that \eqref{EQ:s2t} holds and $c_1 r+c_2 t$ is a 
  $(\sigma_y,\sigma_z)$-remainder for all $c_1,c_2\in \set F$.
\end{theorem}
\begin{proof}
  Since both $r$ and $s$ are $(\sigma_y,\sigma_z)$-remainders, they
  admit the remainder forms \eqref{EQ:rsforms}.  For any $d_i$ from
  $s$, if there exists some $\bar d_k$ from $r$ such that $\bar d_k =
  \sigma_y^\lambda\sigma_z^\mu(d_i)$ for some $\lambda,\mu\in\set Z$,
  then for each integer $j$ with $1\leq j\leq n_i$, one sees from
  Lemma~\ref{LEM:localremainder} that step~2.1.1 correctly finds the
  $g_{ij},h_{ij}$ such that \eqref{EQ:simplefrac} holds. Moreover,
  $\sigma_y^\lambda\sigma_z^\mu(a_{ij})/(\sigma_y^\lambda(b_{ij})\bar d_k^j)$
  is a remainder fraction if $d_i$ is not $(y,z)$-integer
  linear over~$\set F$.  When $d_i$ is $(y,z)$-integer linear over
  $\set F$ of $(\alpha_i,\beta_i)$-type, Lemma~\ref{LEM:bishiftcoprime}
  assures that \eqref{EQ:bishift} holds and $a_{ij}'/(b_{ij}'\bar d_k^j)$
  is a remainder fraction. Note that $d_1,\dots,d_m$ are pairwise
  $(y,z)$-shift inequivalent since $s$ is a $(\sigma_y,\sigma_z)$-remainder.
  Also note that each $d_i$ can only be replaced by some $\bar d_k$
  which is $(y,z)$-shift equivalent to $d_i$ every time the algorithm
  passes through step~2.2. Thus the updated $d_i$ after step~2 remain
  to be $(y,z)$-shift inequivalent to each other. It then follows from
  Definition~\ref{DEF:remainder} that
  $t=\sum_{i=1}^m\sum_{j=1}^{n_i}a_{ij}/(b_{ij}{d_i}^j)$ in step~3
  (with a slight abuse of notations) is a $(\sigma_y,\sigma_z)$-remainder.
  Substituting all equations \eqref{EQ:simplefrac}-\eqref{EQ:bishift}
  into \eqref{EQ:rsforms}, together with \eqref{EQ:integerlinearsummable},
  immediately yields \eqref{EQ:s2t}.
 
  Let $c_1,c_2\in \set F$. Then it remains to prove that $c_1r+c_2 t$
  is a $(\sigma_y,\sigma_z)$-remainder. Notice that for any two
  remainder fractions: $\bar a_{kj}/(\bar b_{kj}\bar d_k^j)$ from $r$
  and $a_{ij}/(b_{ij}d_i^j)$ from $t$ with $\bar d_{k}\nsim_{y,z} d_i$,
  it is readily seen from definition that their any linear combination
  over $\set F$ is again a remainder fraction.  Thus it amounts to
  showing that $c_1\bar a_{kj}/(\bar b_{kj}\bar d_k^j)+c_2a_{ij}/(b_{ij}d_i^j)$
  is a remainder fraction in the case when $\bar d_k\sim_{y,z}d_i$. We
  know from step~2 that in this case $d_i = \bar d_k$, and
  $\phi_{\alpha_i,\beta_i}(c_1\bar a_{kj}/\bar b_{kj}+c_2a_{ij}/b_{ij})$
  is a $\sigma_y$-remainder if $d_i$ is $(y,z)$-integer linear over
  $\set F$ of $(\alpha_i,\beta_i)$-type. Therefore, the theorem is
  concluded by definition.
  \end{proof}

\section{Telescoping via reduction}\label{SEC:rct}
Recall that a telescoper $L$, for a given rational function $f\in \set
F(y,z)$, is a nonzero operator in $\set F[\shift_x]$ such that $L(f)$
is $(\sigma_y,\sigma_z)$-summable.  For discrete trivariate rational
functions, telescopers do not always exist.  Recently, a criterion for
determining the existence of telescopers in this case was presented in
the work \cite{CHLW2016}.
In order to adapt it into our setting, we will consider
primitive parts of polynomials in $\set F[y]$. 
Let $p\in\set F[y]$ be of the form $p = \sum_{i=0}^d(a_i/b)y^i$ for
$d\in\set N$ and $a_i,b\in\set K[x]$ with $b\neq 0$. Then the 
{\em content} $\cont_y(p)$ of $p$ with respect to $y$ is defined as 
$\cont_y(p) = \gcd(a_0,\dots,a_d)/b\in\set F$, and the corresponding 
{\em primitive part} $\prim_y(p)=p/\cont_y(p)$. For example,
by letting $p = 3xy-9x/(x+1)\in\set F[y]$, we have
$\cont_y(p) = 3x/(x+1)\in\set F$ and $\prim_y(p)=(x+1)y-3\in\set K[x,y]$.
Evidently, $\prim_{y}(p)$ is a polynomial in $\set K[x,y]$
whose coefficients with respect
to $y$ have no nonconstant common divisors in $\set K[x]$.

 We summarize the existence criterion for telescopers from \cite{CHLW2016} 
 in the following theorem. 
\begin{theorem}[Existence criterion]\label{THM:existence}
  Let $f$ be a rational function in $\set F(y,z)$. Assume that
  applying the bivariate Abramov reduction to $f$ yields
  \eqref{EQ:biabred}, where $g,h,r\in \set F(y,z)$ and $r$ is a
  $(\sigma_y,\sigma_z)$-remainder with the remainder form
  \eqref{EQ:remainder}. Then $f$ has a telescoper if and only if for
  each $1\leq i\leq m$ and $1\leq j\leq n_i$,
  \begin{itemize}
  \item[(i)] there exists a positive integer $\xi_i$ such that
    $\sigma_x^{\xi_i}(d_i) = \sigma_y^{\zeta_i}\sigma_z^{\eta_i}(d_i)$
    for some integers $\zeta_i,\eta_i$,
  \item[(ii)] and $b_{ij}$ is $(x,y)$-integer linear over~$\set K$, in
    particular, $\sigma_x^{\xi_i}(\prim_y(b_{ij})) =
    \sigma_y^{\zeta_i}(\prim_y(b_{ij}))$ if $d_i$ is not $(y,z)$-integer
    linear over $\set F$.
  \end{itemize}
\end{theorem}
Algorithms for checking the conditions (i)-(ii) were also described in
the same paper \cite{CHLW2016}.  With termination guaranteed by the above
criterion, we now use the bivariate Abramov reduction to develop a
creative telescoping algorithm in the spirit of the general
reduction-based approach.

\medskip\noindent{\bf Algorithm ReductionCT.} Given a rational function $f\in
\set F(y,z)$, compute a  minimal telescoper $L\in \set F[\shift_x]$
for~$f$ and a corresponding certificate $(g,h)\in \set F(y,z)^2$
when telescopers exist.

\begin{enumerate}
\item Apply the bivariate Abramov reduction to $f$ to find
  $g_0,h_0\in \set F(y,z)$ and a $(\sigma_y,\sigma_z)$-remainder
  $r_0\in \set F(y,z)$ such that
  \begin{equation}\label{EQ:0red}
    f = (\shift_y-1)(g_0)+(\shift_z-1)(h_0)+r_0.
  \end{equation}
\item If $r_0=0$ then set $L = 1$, $(g,h)=(g_0,h_0)$ and return.
\item If conditions~(i)-(ii) in Theorem~\ref{THM:existence} are not
  satisfied, then return \lq\lq No telescopers exist\rq\rq.
\item Set $R = u_0r_0$, where $u_0$ is an indeterminate.\\[1ex]
  For $\ell = 1, 2, \dots $ do
  \begin{itemize}

  \item[4.1] Apply the remainder linearization to
    $\sigma_x(r_{\ell-1})$ with respect to $R$ to find
    $g_\ell,h_\ell\in \set F(y,z)$ and a
    $(\sigma_y,\sigma_z)$-remainder $r_\ell\in \set F(y,z)$ such that
    \begin{equation}\label{EQ:ellthred2}
      \sigma_x(r_{\ell-1}) = (\shift_y-1)(g_\ell)+(\shift_z-1)(h_\ell)+r_\ell,
    \end{equation}
    and that $R+u_\ell r_\ell$ is a $(\sigma_y,\sigma_z)$-remainder,
    where $u_\ell$ is an indeterminate.
    
  \item[4.2] Update $R = R+u_\ell r_\ell$ and update
    $g_\ell=g_\ell+\sigma_x(g_{\ell-1})$, $h_\ell=h_\ell+\sigma_x(h_\ell)$
    so that
    \begin{equation}\label{EQ:ellthred}
      \sigma_x^\ell(f) = (\shift_y-1)(g_\ell)+(\shift_z-1)(h_\ell)+r_\ell.
    \end{equation}
    
  \item[4.3] If there exist nontrivial $c_0,\dots, c_\ell\in \set F$
    such that $R\mid_{u_i=c_i}=0$, then set $L=\sum_{i=0}^\ell c_i\shift_x^i$
    and $(g,h) = (\sum_{i=0}^\ell c_i g_i, \sum_{i=0}^\ell c_i h_i)$,
    and return.
  \end{itemize}
\end{enumerate}
\begin{theorem}\label{THM:rct}
  Let $f$ be a rational function in $\set F(y,z)$. Then the algorithm
  {\bf ReductionCT} terminates and returns a minimal telescoper for
  $f$ when such a telescoper exists.
\end{theorem}
\begin{proof}
  By Theorems~\ref{THM:biabred} and \ref{THM:existence}, steps~2-3 are
  correct. Because $r_0$ is a $(\sigma_y,\sigma_z)$-remainder, so is its shift
  $\sigma_x(r_0)$. By
  Theorem~\ref{THM:remsum}, step~4.1 correctly finds $g_1,h_1\in \set
  F(y,z)$ and a $(\sigma_y,\sigma_z)$-remainder $r_1\in \set F(y,z)$
  such that \eqref{EQ:ellthred2} holds for $\ell=1$ and
  $R+u_1r_1=u_0r_0+u_1r_1$ is a $(\sigma_y,\sigma_z)$-remainder for
  all $u_0, u_1\in \set F$.  Applying $\sigma_x$ to both sides of
  \eqref{EQ:0red}, together with step~4.1, one sees that step~4.2
  gives \eqref{EQ:ellthred} for $\ell=1$.  The correctness of step~4.2
  for each iteration of the loop in step~4 then follows by induction
  on $\ell$.
  
  If $f$ does not have a telescoper then the algorithm halts after
  step~3. Otherwise, telescopers for $f$ exist
  by~Theorem~\ref{THM:existence}. Let $L= \sum_{\ell=0}^\rho c_\ell
  \shift_x^\ell\in \set F[\shift_x]$ be a telescoper for $f$ of
  minimal order. Then $c_\rho\neq 0$ and by~\eqref{EQ:ellthred},
  applying $L$ to $f$ gives
  \[L(f)=\sum_{\ell=0}^\rho c_\ell\sigma_x^\ell(f)
  = (\shift_y-1)\left(\sum_{\ell=0}^\rho c_\ell g_\ell\right)
  +(\shift_z-1)\left(\sum_{\ell=0}^\rho c_\ell h_\ell\right)
  +\sum_{\ell=0}^\rho c_\ell r_\ell.\]
  Notice that $\sum_{\ell=0}^\rho c_\ell r_\ell$ is a
  $(\sigma_y,\sigma_z)$-remainder by step~4.1. It follows from
  Theorem~\ref{THM:biabred} that $L(f)$ is
  $(\sigma_y,\sigma_z)$-summable, namely $L$ is a telescoper for $f$,
  if and only if $\sum_{\ell=0}^\rho c_\ell r_\ell = 0$. This implies
  that the linear system over $\set F$ with unknowns $u_\ell$ obtained by equating
  $\sum_{\ell=0}^\rho u_\ell r_\ell$ to zero has a nontrivial
  solution, which yields a telescoper of minimal order. The algorithm
  thus terminates.
\end{proof}

In what follows, we describe an alternative way, in addition to the
above algorithm, for creative telescoping in our trivariate rational
setting. As such, we need the notion of least common left multiples.
Recall that an operator $L\in \set
F[\shift_x]$ is a {\em common left multiple} of operators $L_1,\dots,
L_m\in \set F[\shift_x]$ if there exist operators $L_1',\dots,L_m'\in
\set F[\shift_x]$ such that $L = L'_1L_1=\cdots = L_m'L_m$. Amongst
all such common left multiples, the monic one of lowest possible
degree with respect to $\shift_x$ is called the
{\em least common left multiple}. Many efficient algorithms for
computing the least common left multiple of operators are available;
see~\cite{ALL2005} and the references therein.

The following lemma is an immediate extension of \cite[Theorem~2]{Le2003a}
to the trivariate case, and thus we omit the proof.
\begin{lemma}\label{LEM:lclm}
  Let $r = r_1+\cdots+r_m$ with $r_i\in \set F(y,z)$ and let
  $L_1,\dots,L_m\in \set F[\emph\shift_x]$ be the respective minimal
  telescopers for $r_1, \dots, r_m$. Then the least common left
  multiple $L$ of the $L_i$ is a telescoper for $r$. Moreover, if any
  telescoper for $r$ is also a telescoper for each $r_i$ with $1\leq
  i\leq m$, then $L$ is a minimal telescoper for $r$.
\end{lemma}
The following proposition shows that the least common multiple gives a
minimal telescoper for the given sum provided that the denominators of
distinct summands are comprised of distinct $(x,y,z)$-shift equivalence
classes.
\begin{proposition}\label{PROP:lclm}
  Let $r\in \set F(y,z)$ be a rational function of the form
  \begin{equation*}
    r = r_1+r_2+\cdots+r_m,
  \end{equation*}
  where $r_i = a_i/d_i$ with $a_i,d_i\in \set F[y,z]$ satisfying the conditions
  \begin{itemize}
  \item[(i)] $\deg_z(a_i)<\deg_z(d_i)$;
  \item[(ii)] any monic irreducible factor of $d_i$ of positive degree
    in $z$ is $(x,y,z)$-shift inequivalent to all factors of $d_{i'}$
    whenever $1\leq i,i'\leq m$ and $i\neq i'$.
  \end{itemize}
  Let $L_1,\dots, L_m\in \set F[\emph\shift_x]$ be the respective minimal
  telescopers for $r_1,\dots, r_m$. Then the least common left
  multiple $L$ of the $L_i$ is a minimal telescoper for $r$. Moreover,
  for each $1\leq i\leq m$, let $(g_{i},h_{i})$ be a corresponding
  certificate for $L_{i}$ and let $L'_{i}\in \set F[\emph\shift_x]$ be the
  cofactor of $L_i$ so that $L = L'_{i}L_{i}$. Then
  \[\left(\sum_{i=1}^mL'_{i}(g_{i}),\sum_{i=1}^mL'_{i}(h_{i})\right)\]
  is a corresponding certificate for $L$.
\end{proposition}
\begin{proof}
  Let $\tilde L\in \set F[\shift_x]$ be a telescoper for $r$. In order
  to show the first assertion, by Lemma~\ref{LEM:lclm}, it suffices to
  verify that $\tilde L$ is also a telescoper for each $r_{i}$ with
  $1\leq i\leq m$. Notice that the application of a nonzero operator
  from $\set F[\shift_x]$ does not change the $(x,y,z)$-shift
  equivalence classes, with representatives being monic irreducible
  polynomials of positive degrees in $z$, that appear in a given
  polynomial in $\set F[y,z]$. Hence condition~(ii) remains valid when
  $d_i$ and $d_{i'}$ are replaced by $\tilde L(d_i)$ and $\tilde
  L(d_{i'})$, respectively.  It then follows that any two monic
  irreducible factors of positive degrees in $z$ from distinct $d_{i}$
  are $(y,z)$-shift inequivalent to each other. By the definition of
  telescopers, $\tilde L(r)$ is $(\sigma_y,\sigma_z)$-summable, and
  then so is each $\tilde L(r_{i})$ according to
  \cite[Lemma~3.2]{HoWa2015}.  This implies that $\tilde L$ is indeed
  a telescoper for each $r_{i}$. The second assertion follows by
  observing that $(\shift_y-1)$ and $(\shift_z-1)$ both commute with
  operators from $\set F[\shift_x]$.
\end{proof}

The above proposition suggests a natural approach to construct a
minimal telescoper for a given rational function. More precisely, let
$f\in\set F(y,z)$ and assume that applying the bivariate Abramov
reduction to $f$ yields \eqref{EQ:biabred} with $r$ admitting the
remainder form \eqref{EQ:remainder}.  The approach proceeds by
separately taking each $\sum_{j=1}^{n_i}a_{ij}/(b_{ij}d_i^j)$ in
\eqref{EQ:remainder} as the basic case and computes its own minimal
telescoper $L_i\in\set F[\shift_x]$ using the algorithm {\bf ReductionCT},
and then returns the least common left multiple $L$ of all $L_i$ as
the output.  By Proposition~\ref{PROP:lclm}, such an $L$ gives a
minimal telescoper for $r$ (and thus for $f$). We refer to this
approach as the LCLM version of our algorithm {\bf ReductionCT}.

\subsection{Examples}

\begin{example}\label{EX:f1}
  Consider the rational function $f_1$ given in
  Example~\ref{EX:example}.  Note that $f_1$ is a remainder fraction
  and satisfies conditions~(i)-(ii) in Theorem~\ref{THM:existence}. So
  telescopers for $f_1$ exist. Applying the algorithm {\bf ReductionCT}
  to $f_1$, we obtain in step~4 that
  \[\sigma_x^\ell(f_1) = (\shift_y-1)(g_\ell)+(\shift_z-1)(h_\ell)+r_\ell
  \quad\text{for}\ \ell=0,1, 2,\]
  where
  \begin{align*}
    r_0=f_1, \quad
    r_1 = \frac{(x+1)^2(z-1)+1}{(x+y)(x+z)^2+1},\quad
    r_2 = \frac{(x+2)^2(z-2)+1}{(x+y)(x+z)^2+1}
  \end{align*}
  and $g_\ell,h_\ell\in \set F(y,z)$ are not displayed here to keep
  things neat. By finding an $\set F$-linear dependency among
  $r_0,r_1,r_2$, we see that
  \begin{align*}
    L_1&=(x^4+2x^3+x^2+2x+1)\emph\shift_x^2
    -2(x^4+4x^3+4x^2+2x+2)\emph\shift_x+(x^4+6x^3+13x^2+14x+7)
  \end{align*}
  is a minimal telescoper for $f_1$.
\end{example}

\begin{example}\label{EX:f3}
  Consider the rational function $f_2$ given in
  Example~\ref{EX:example}, which can be decomposed as
  \[f_2 = (\shift_y-1)(g_0)+(\shift_z-1)(h_0)+\underbrace{
    \frac{\frac13x(2y+3z)+\frac23x^2+1}{(x+y)((x+2y+3z)^2+1)}}_{r_0}
  \quad \text{for some} \ g_0,h_0\in \set F(y,z).\]
  Note that $r_0$ is a remainder fraction and satisfies
  conditions~(i)-(ii) in Theorem~\ref{THM:existence}. Thus telescopers
  for $f_2$ exist. Applying the algorithm {\bf ReductionCT} to $f_2$,
  we obtain in step~4 that
  \[\sigma_x^\ell(f_2)=(\shift_y-1)(g_\ell)+(\shift_z-1)(h_\ell)+r_\ell
  \quad\text{for}\ \ell=0,1,\dots,6,\]
  where $g_\ell,h_\ell\in \set F(y,z)$ are again not displayed due to
  the large sizes, and
  \begin{align*}
    \begin{array}{lll}
      r_1 = \frac{(\frac13x+1)(2y+3z)+\frac23x^2+x+\frac43}
      {(x+y+2)((x+2y+3z)^2+1)},&~
      r_2 = \frac{(\frac13x+\frac23)(2y+3z)+\frac23x^2+2x+\frac73}
      {(x+y+4)((x+2y+3z)^2+1)},&~
      r_3 = \frac{(\frac13x+1)(2y+3z)+\frac23x^2+4}
      {(x+y)((x+2y+3z)^2+1)},\\[1.5ex]
      r_4 = \frac{(\frac13x+\frac43)(2y+3z)+\frac23x^2+4x+\frac{19}3}
      {(x+y+2)((x+2y+3z)^2+1)},&~
      r_5 = \frac{(\frac13x+\frac53)(2y+3z)+\frac23x^2+5x+\frac{28}3}
      {(x+y+4)((x+2y+3z)^2+1)},&~
      r_6 = \frac{(\frac13x+2)(2y+3z)+\frac23x^2+6x+13}
      {(x+y)((x+2y+3z)^2+1)}.
    \end{array}
  \end{align*}
  Then one finds an $\set F$-linear dependency among $r_0,r_3,r_6$
  which yields a minimal telescoper
  \[L_2 = (x^2+3x-3)\emph\shift_x^6-2(x^2+6x-3)\emph\shift_x^3+x^2+9x+15.\]
\end{example}

The following illustrates the result of  Proposition \ref{PROP:lclm}.

\begin{example}\label{EX:rct}
  Consider the same rational function $f$ as in
  Example~\ref{EX:example}. Then we know that $f_3$ is
  $(\sigma_y,\sigma_z)$-summable. Thus $L_3 =1$ is a minimal
  telescoper for $f_3$. Let $L_1,L_2\in \set F[\shift_x]$ be the
  operators computed in Examples~\ref{EX:f1}-\ref{EX:f3}. It then
  follows that the least common left multiple $L$ of
  $\{L_1,L_2,L_3\}$, given by
  \begin{align*}
  L &= \emph\shift_x^{8}
  -\tfrac{2(x^2+5x+1)(3x^2+24x+31)}{(x^2+7x+7)(3x^2+21x+19)}\emph\shift_x^7
  +\tfrac{(x^2+3x-3)(3x^2+27x+43)}{(x^2+7x+7)(3x^2+21x+19)}\emph\shift_x^6
  -\tfrac{2(x^2+10x+13)}{x^2+7x+7}\emph\shift_x^5\\[1ex]
  & \quad
  +\tfrac{4(3x^2+24x+31)(x^2+8x+4)}{(x^2+7x+7)(3x^2+21x+19)}\emph\shift_x^4
  -\tfrac{2(x^2+6x-3)(3x^2+27x+43)}{(x^2+7x+7)(3x^2+21x+19)}\emph\shift_x^3
  +\tfrac{x^2+13x+37}{x^2+7x+7}\emph\shift_x^2\\[1ex]
  &\quad
  -\tfrac{2(x^2+11x+25)(3x^2+24x+31)}{(x^2+7x+7)(3x^2+21x+19)}\emph\shift_x
  +\tfrac{(x^2+9x+15)(3x^2+27x+43)}{(x^2+7x+7)(3x^2+21x+19)},
  \end{align*}
  is a telescoper for $f$. On the other hand, by directly applying the
  algorithm {\bf ReductionCT} to $f$, one sees that $L$ is in fact a
  minimal telescoper for $f$.
\end{example}

\subsection{Efficiency considerations}

The efficiency of Algorithm {\bf ReductionCT} can be enhanced by incorporating  
two modifications in the algorithm.

\bigskip

\noindent
{\bf Simplification of step 4.1}

\bigskip
For each iteration of the loop in step~4, rather than using the overall 
$(\sigma_y,\sigma_z)$-remainder $R=\sum_{k=0}^{\ell-1}u_k r_k$ in step~4.1, we 
can apply the remainder linearization to the shift value $\sigma_x(r_{\ell-1})$ 
with respect to the initial $(\sigma_y,\sigma_z)$-remainder $r_0$ only.  This is
sufficient as, for any $(\sigma_y,\sigma_z)$-remainder $r_\ell$
of $\sigma_x(r_{\ell-1})$ with $\ell\geq 1$, if $r_0+r_\ell$ is a
$(\sigma_y,\sigma_z)$-remainder then so is $R+u_\ell r_\ell$, provided
that the algorithm proceeds in the described iterative
fashion. 

The intuition for this simplification is as follows. Notice that
if the algorithm continues after passing through step~3 then $r_0\neq 0$. 
Since distinct $(y,z)$-shift equivalence classes can be tackled separately, 
we restrict ourselves to the case where the denominator of $r_0$ is of the 
form
\[d\,\sigma_x^{i_1}(d) \cdots\sigma_x^{i_m}(d)\]
with $d\in \set F[y,z]$ being monic, irreducible and of positive
degree in $z$, $i_1,\dots,i_m$ being distinct positive integers such
that $d, \sigma_x^{i_1}(d),\dots, \sigma_x^{i_m}(d)$ are $(y,z)$-shift
inequivalent to each other. For simplicity, we call
$(0,i_1,\dots,i_m)$ the $x$-shift exponent sequence of $d$ in
$r_0$. By Theorem~\ref{THM:existence}, there exists a positive integer
$\xi$ such that $\sigma_x^\xi(d)\sim_{y,z} d$ and so we let $\xi$ be the
smallest one with such a property. Then there are only $\xi$ many
$(y,z)$-shift equivalence classes produced by shifting $d$ with
respect to $x$, with $d,\sigma_x(d), \dots,\sigma_x^{\xi-1}(d)$ as
respective representatives. Without loss of generality, we further
assume that $0<i_1<\cdots<i_m<\xi$. For $\ell\geq 1$, let $r_\ell$ be
the output of the remainder linearization when applied to
$\sigma_x(r_{\ell-1})$ with respect to $r_0$. By induction on $\ell$,
one sees that the $x$-shift exponent sequence of $d$ in $r_\ell$ is
given by
\[(\ell, i_1+\ell,\dots, i_m+\ell) \mod \xi,\]
whose entries form an $(m+1)$-subset of $\{0,1,\dots,\xi-1\}$. It thus
follows from Definition~\ref{DEF:remainder} that $R+u_\ell r_\ell$ is
also a $(\sigma_y,\sigma_z)$-remainder.

\bigskip
\noindent
{\bf Simplification of step 4.3}

\bigskip
Our second modification is in step~4.3, where we first derive
from $R=0$ the individual equation for each remainder fraction $a/(b\, d^j)$
appearing in the remainder form of $R$, and then build a linear
system over $\set F$ from the coefficients of the numerator of the
equation with respect to $y$ and $Z = \alpha y+ \beta z$, instead of
$y$ and $z$, in the case where $d$ is $(y,z)$-integer linear of
$(\alpha,\beta)$-type. Notice that $R = u_0r_0+u_1r_1+\dots+u_\ell
r_\ell$ at the stage of step~4.3. Let $d_1,\dots, d_m$ be all monic
irreducible polynomials of positive degrees in $z$ that appear in the
denominator of $R$, with multiplicities $n_1,\dots, n_m$,
respectively. For $1\leq i\leq m$, $1\leq j\leq n_i$ and $0\leq k\leq
\ell$, let $a_{ij}^{(k)}\in \set F[y,z]$ and $b_{ij}^{(k)}\in \set
F[y]$ be such that $a_{ij}^{(k)}/(b_{ij}^{(k)}d_i^j)$ is a remainder
fraction appearing in the remainder form of $r_k$. By coprimeness
among the $d_i$, one gets that
\[
R=0\quad\Longleftrightarrow\quad
\sum_{k=0}^\ell u_k\cdot\frac{a_{ij}^{(k)}}{b_{ij}^{(k)}} = 0
\quad\text{for all}\ i=1,\dots, m\text{ and }j=1,\dots, n_i.
\] 
If $d_i$ is $(y,z)$-integer linear of $(\alpha_i,\beta_i)$-type, then
$\phi_{\alpha_i,\beta_i}(a_{ij}^{(k)}/b_{ij}^{(k)})$ is a $\sigma_y$-remainder
with $\phi_{\alpha_i,\beta_i}$ being the map for $(\alpha_i,\beta_i)$-shift reduction.
By letting $\tilde a_{ij}^{(k)}=\phi_{\alpha_i,\beta_i}(a_{ij}^{(k)})\in\set F[y,z]$,
one sees from definition that $\deg_y(\tilde a_{ij}^{(k)}) <
\deg_y(\phi_{\alpha_i,\beta_i}(b_{ij}^{(k)}))=\deg_y(b_{ij}^{(k)})$
and $a_{ij}^{(k)} = \phi_{\alpha_i,\beta_i}^{-1}(\tilde a_{ij}^{(k)})
=\tilde a_{ij}^{(k)}(\beta_i^{-1}y,\beta_iz+\alpha_iy)$. It follows
that every $a_{ij}^{(k)}$ can be viewed as a polynomial in $Z_i=
\alpha_i y+ \beta_i z$ with coefficients all having degrees in $y$
less than $\deg_y(b_{ij}^{(k)})$. In this case, rather than naively
considering the coefficients with respect to $y$ and $z$, we instead
force all the coefficients with respect to $y$ and $Z_i$ of the
numerator of $\sum_{k=0}^\ell u_k\cdot(a_{ij}^{(k)}/b_{ij}^{(k)})$ to
zero. This way ensures that the resulting linear system over $\set F$
typically has smaller size than the naive one.

\section{Implementation and timings}\label{SEC:test}

We have implemented our new algorithm {\bf ReductionCT} in the
computer algebra system {\sc Maple~2018}. Our implementation includes
the two enhancements to step~4 discussed in the previous subsection.
In order to get an idea about the efficiency of our algorithm, we
applied our implementation to certain examples and tabulated their
runtime in this section. All timings were measured in seconds on a
Linux computer with 128GB RAM and fifteen 1.2GHz Dual core
processors. The computations for the experiments did not use any
parallelism.

We considered trivariate rational functions of the form
\begin{equation}\label{EQ:test}
  f(x,y,z) = \frac{a(x,y,z)}{d_1(x,y,z)\cdot d_2(x,y,z)},
\end{equation}
where
\begin{itemize}
\item $a\in \set Z[x,y,z]$ of total degree $m\geq 0$ and max-norm
  $||a||_\infty\leq 5$, in other words, the maximal absolute value of
  the coefficients of $a$ with respect to $x,y,z$ are no more than 5;

\item $d_i = p_i\cdot\sigma_x^{\xi}(p_i)$ with $p_1=P_1(\xi y-\zeta
  x,\xi z+\zeta x)$ and $p_2=P_2(\zeta x+\xi y+2\xi z)$ for two
  nonzero integers $\xi,\zeta$ and two integer polynomials
  $P_1(y,z)\in \set Z[y,z]$, $P_2(z)\in \set Z[z]$, both of which have
  total degree $n>0$ and max-norm no more than 5.
\end{itemize}

For a selection of random rational functions of this type for
different choices of $(m,n,\xi,\zeta)$, Table~\ref{TAB:test} collects
the timings of four variants of the algorithm {\bf ReductionCT} from
Section~\ref{SEC:rct}. For the column RCT$_{1}$, we computed both the
telescoper and the certificate, and for the column RCT$_{2}$ only the
telescoper is computed. The difference between these two variants
mainly lies in the time used to bring the certificate to a common
denominator. When it is acceptable to keep the certificate as an
unnormalized linear combination of rational functions, the timings are
virtually the same as for RCT$_2$. For columns RCTLM$_1$ and
RCTLM$_2$, we perform the same functionality as RCT$_1$ and RCT$_2$
but using the LCLM version of the algorithm {\bf ReductionCT}. Note
that the computation of the least common left multiples therein was
accomplished by the built-in Maple command {\sf OreTools[LCM]['left']}.
We remark that the performance of the LCLM version of the algorithm
{\bf ReductionCT} deteriorates for larger examples, especially when there
are many shift equivalence classes in the denominator of the input
rational function or the order of a minimal telescoper is relatively high.
\begin{table}[ht]
  \centering
  \begin{tabular}{l|rr|rr|c}
    $(m,n,\xi,\zeta)$ & RCT$_1$ & RCT$_2$
    & RCTLM$_1$ & RCTLM$_2$ & order\\\hline
    (1, 1, 1, 1) & 0.196 & 0.098 & 0.220 & 0.110 & 1\\
    (1, 1, 1, 5) & 7.319 & 0.123 & 9.483 & 0.123 & 1\\
    (1, 1, 1, 9) & 105.548 & 0.121 & 104.514 & 0.125 & 1\\
    (1, 1, 1, 13) & 2586.295 & 0.136 & 3078.043 & 0.126 & 1\\
    (1, 1, 1, 3) & 0.574 & 0.097 & 0.712 & 0.104 & 1\\
    (1, 2, 1, 3) & 17.812 & 0.256 & 17.299 & 0.263 & 1\\
    (1, 3, 1, 3) & 266.206 & 1.999 & 220.209 & 1.997 & 1\\
    (1, 4, 1, 3) & 2838.827 & 37.358 & 3039.199 & 30.547 & 1\\
    (1, 5, 1, 3) & 19403.916 & 1074.295 & 18309.000 & 1119.393 & 1\\
    (2, 3, 1, 3) & 31678.706 & 2.540 & 15825.876 & 2.224 & 3\\
    (3, 3, 1, 3) & 44243.254 & 5.378 & 16869.097 & 4.295 & 3\\
    (3, 2, 1, 3) & 710.810 & 0.492 & 670.501 & 0.487 & 3\\
    (3, 2, 2, 3) & 1314.809 & 0.701 & 941.009 & 0.756 & 6\\
    (3, 2, 4, 3) & 1558.440 & 1.525 & 1121.624 & 1.550 & 12\\
    (3, 2, 8, 3) & 1878.424 & 4.215 & 986.017 & 4.245 & 24\\
    (3, 2, 16, 3) & 2800.050 & 21.136 & 1317.603 & 38.504 & 48\\\hline
  \end{tabular}
  \caption{Timings for four variants of the algorithm
    {\bf ReductionCT}.}\label{TAB:test}
\end{table}

We have also compared our procedures with the two Mathematica packages:
{\sf HolonomicFunctions} by Koutschan \cite{Kout2010b} and
{\sf MultiSum}\footnote{We thank the anonymous referee for bringing this package to our attention.}
by Wegschaider (substantially improved by Riese) \cite{Wegs1997,LPR2002}.
The {\sf HolonomicFunctions}, to our best knowledge,
is the most comprehensive implementation in terms of creative telescoping
for holonomic functions (cf.\ \cite[\S 2.2]{Kout2009}) in more than two variables.
There are two commands available in the package for our purpose.
One is called {\bf CreativeTelescoping}, which implements Chyzak's algorithm
\cite{Chyz2000} for single sums and can be applied iteratively to compute
telescopers for trivariate rational functions. The other is called
{\bf FindCreativeTelescoping}, which is based on Koutschan's heuristic
approach \cite{Kout2010a} and constructs the telescoper directly by
guessing the denominators of the certificate, as well as their numerator
degrees, and solving a linear system. The {\sf MultiSum} extends the
multivariate version of \lq\lq Sister Celine's technique\rq\rq\ developed
by Wilf and Zeilberger \cite{WiZe1992a}. The available command in the
package is called {\bf FindRecurrence}, which finds a telescoper and a
corresponding certificate for a given summand only if the structure set,
which is usually not known in advance, is chosen in a clever way. 
The idea employed in the package is to use random parameter substitutions
to quickly rule out useless structure sets, which however requires a priori bounds 
for the shifts involved (see \cite{LPR2002} for further details). 
We remark that\footnote{We thank the anonymous referee for pointing this out.} 
it would be interesting to see in the future if our fully automatic method could 
provide these extra bounds also automatically, and then the combination of 
the two methods might yield even a new fully automatic (and efficient) method.

Experiments suggest a better performance of our algorithm.
For example, for the rational function
\begin{align*}
f = \frac{4x + 2}{(45x + 5y+10z + 47)(45x + 5y+10z + 2)(63x - 5y+2z + 58)(63x - 5y+2z - 5)}
\end{align*}
which was constructed using \eqref{EQ:test} with parameter $(m,n,\xi,\zeta) = (1,1,1,9)$,
our algorithm found a minimal telescoper for $f$ along with its corresponding
certificate in about 3 minutes; while the command {\bf FindRecurrence},
along with a priori bounds $1,9,9$ for the shifts in $x,y,z$, respectively,
accomplished the same job using about 7 minutes, the command
{\bf CreativeTelescoping} took about 4 hours, and the command
{\bf FindCreativeTelescoping} did not finish in reasonable time, which happens
because the guessed denominators are wrong/insufficient, and therefore the
command finds nothing and runs forever. The same phenomenon was observed
for larger examples.

\section{Conclusion and future work}\label{SEC:conclu}
In this paper, we presented a new creative telescoping algorithm for
the class of trivariate rational functions.  The procedure is based on
a bivariate extension of Abramov's reduction method initiated in
\cite{Abra1975}. Our algorithm finds a minimal telescoper for a given
trivariate rational function without also needing to compute an
associated certificate. A Maple implementation indicates the
efficiency of our algorithm. As a next step, we are going to
investigate the theoretical complexity of our algorithm to see if it
matches with the practical performance, something briefly alluded to
in the introduction.

We are interested in the more general and important problem of computing
hypergeometric multiple summations or proving identities which involve
such summations. A function $f(x,y_1,\dots, y_n)$ is called a
multivariate {\em hypergeometric term} if the quotients
\[\frac{f(x+1,y_1,\dots, y_n)}{f(x,y_1,\dots,y_n)},
\frac{f(x,y_1+1,\dots, y_n)}{f(x,y_1,\dots,y_n)}, \dots,
\frac{f(x,y_1,\dots, y_n+1)}{f(x,y_1,\dots,y_n)}\]
are all rational functions in $x,y_1,\dots, y_n$.  The problem of
hypergeometric multiple summations tends to appear more often than the
rational case, particularly in combinatorics \cite{AnPa1993,BLS2017},
and it is also more challenging.

Since a large percent of hypergeometric terms falls into the class of
holonomic functions, the problem of hypergeometric multiple summations
can also be considered in a more general framework of multivariate
holonomic functions.  In this context, several creative telescoping
approaches have already been developed in
\cite{Zeil1990b,Taka1990,ChSa1998,Chyz2000,Kout2010a,BRS2018}.  The algorithms
in the first three papers are based on elimination and suffer from the
disadvantage of inefficiency in practice. The algorithm in
\cite{Chyz2000}, also known as Chyzak's algorithm, deals with single
sums (and single integrals) and can only be used to solve multiple
ones in an iterative manner. A fast but heuristic approach was given in
\cite{Kout2010a} in order to eliminate the bottleneck in Chyzak's
algorithm of solving a coupled first-order system.  This approach
generalizes to multiple sums (and multiple integrals).  We refer to
\cite{Kout2009} for a detailed and excellent exposition of these
approaches.
The work in \cite{BRS2018} describes even a general framework that unities
the difference ring and the holonomic approach.
We remark that all these approaches find the telescoper
and the certificate simultaneously, with the exception of Takayama's
algorithm in \cite{Taka1990} where natural boundaries have to be
assured a priori.
Note also that holonomicity is a sufficient but not necessary
condition for the applicability of creative telescoping applied to
hypergeometric terms (cf.~\cite{Abra2003,CHLW2016}).

Restricted to the hypergeometric setting, partial solutions for the
problem of multiple summations were proposed in \cite{CHM2006} and
\cite{BLS2017}. In the former paper, the authors presented a heuristic
method to find telescopers for trivariate hypergeometric terms,
through which they also managed to prove certain famous hypergeometric
double summation identities. In the latter paper, the authors mainly
focused on a subclass of hypergeometric summations -- multiple
binomial sums. They first showed that the generating function of a
given multiple binomial sum is always the diagonal of a rational
function and vice versa. They then constructed a differential equation
for the diagonal by a reduction-based telescoping approach. Finally
the differential equation is translated back into a recurrence
relation satisfied by the given binomial sum. In the future, we hope
to explore this topic further and aim at developing a complete
reduction-based telescoping algorithm for hypergeometric terms in
three or more variables.

\section*{Acknowledgments}
We would like to express our gratitude to Christoph Koutschan for useful instructions
and insightful remarks on his package.
We also would like to thank the anonymous referee for many helpful
and constructive suggestions.
Most of the work presented in this paper was carried out while
H.\ Huang was a postdoctoral fellow at the University of Waterloo.
S.\ Chen was partially supported by the NSFC grants (No.\ 11871067, 12288201) and
the Fund of the Youth Innovation Promotion Association, CAS (2018001).
Q.-H.\ Hou was supported by the NSFC grant (No.\ 11921001).
H.\ Huang and G.\ Labahn were supported by the Natural Sciences and
Engineering Research Council (NSERC) Canada (No. NSERC RGPIN-2020-04276). H.\ Huang was also
supported by the NSFC grant (No.\ 12101105) and
the Fundamental Research Funds for the Central Universities (No. DUT20RC(3)073).
R.-H.\ Wang was supported by the NSFC grants (No.\ 12101449, 11871067) and 
the Natural Science Foundation of Tianjin, China (No.\ 19JCQNJC14500).

\bibliographystyle{plain}

\begin{thebibliography}{10}

\bibitem{Abra1975}
Sergei~A. Abramov.
\newblock The rational component of the solution of a first-order linear
  recurrence relation with a rational right side.
\newblock {\em USSR Comput. Math. Math. Phys.}, 15(4):216--221, 1975.

\bibitem{Abra1989}
Sergei~A. Abramov.
\newblock Rational solutions of linear differential and difference equations
  with polynomial coefficients.
\newblock {\em USSR Comput. Math. Math. Phys.}, 29(6):7--12, 1989.
\newblock Transl. from \v{Z}h. {V}y\v{c}isl. {M}at. i {M}at. {F}iz. 29, pp.
  1611-1620, 1989.

\bibitem{Abra1995a}
Sergei~A. Abramov.
\newblock Indefinite sums of rational functions.
\newblock In {\em Proceedings of {ISSAC}'95}, pages 303--308. ACM, New York,
  1995.

\bibitem{Abra1995b}
Sergei~A. Abramov.
\newblock Rational solutions of linear difference and $q$-difference equations
  with polynomial coefficients.
\newblock In {\em Proceedings of {ISSAC}'95}, pages 285--289. ACM, New York,
  1995.

\bibitem{Abra2003}
Sergei~A. Abramov.
\newblock When does {Z}eilberger's algorithm succeed?
\newblock {\em Adv. in Appl. Math.}, 30(3):424--441, 2003.

\bibitem{ABPS2021}
Sergei~A. Abramov, Manuel Bronstein, Marko Petkov\v{s}ek, and Carsten
  Schneider.
\newblock On rational and hypergeometric solutions of linear ordinary
  difference equations in {$\Pi\Sigma^\ast$}-field extensions.
\newblock {\em J. Symbolic Comput.}, 107:23--66, 2021.

\bibitem{AbLe2002}
Sergei~A. Abramov and Ha~Q. Le.
\newblock A criterion for the applicability of {Z}eilberger's algorithm to
  rational functions.
\newblock {\em Discrete Math.}, 259(1-3):1--17, 2002.

\bibitem{ALL2005}
Sergei~A. Abramov, Ha~Q. Le, and Ziming Li.
\newblock Univariate {O}re polynomial rings in computer algebra.
\newblock {\em J. Math. Sci.}, 131(5):5885--5903, 2005.

\bibitem{APP1998}
Sergei~A. Abramov, Peter Paule, and Marko Petkov{\v{s}}ek.
\newblock {$q$}-{H}ypergeometric solutions of {$q$}-difference equations.
\newblock {\em Discrete Math.}, 180(1-3):3--22, 1998.

\bibitem{AbPe2001a}
Sergei~A. Abramov and Marko Petkov{\v{s}}ek.
\newblock Minimal decomposition of indefinite hypergeometric sums.
\newblock In {\em {P}roceedings of {ISSAC}'01}, pages 7--14. ACM, New York,
  2001.

\bibitem{AnPa1993}
George~E. Andrews and Peter Paule.
\newblock Some questions concerning computer-generated proofs of a binomial
  double-sum identity.
\newblock {\em J. Symbolic Comput.}, 16(2):147--151, 1993.

\bibitem{BaPe1999}
Andrej Bauer and Marko Petkov\v{s}ek.
\newblock Multibasic and mixed hypergeometric {G}osper-type algorithms.
\newblock {\em J. Symbolic Comput.}, 28(4-5):711--736, 1999.

\bibitem{BRS2018}
Johannes Bl\"{u}mlein, Mark Round, and Carsten Schneider.
\newblock Refined holonomic summation algorithms in particle physics.
\newblock In {\em Advances in computer algebra}, volume 226 of {\em Springer
  Proc. Math. Stat.}, pages 51--91. Springer, Cham, 2018.

\bibitem{BCCL2010}
Alin Bostan, Shaoshi Chen, Fr{\'e}d{\'e}ric Chyzak, and Ziming Li.
\newblock Complexity of creative telescoping for bivariate rational functions.
\newblock In {\em {P}roceedings of {ISSAC}'10}, pages 203--210. ACM, New York,
  2010.

\bibitem{BCCLX2013}
Alin Bostan, Shaoshi Chen, Fr{\'e}d{\'e}ric Chyzak, Ziming Li, and Guoce Xin.
\newblock Hermite reduction and creative telescoping for hyperexponential
  functions.
\newblock In {\em {P}roceedings of {ISSAC}'13}, pages 77--84. ACM, New York,
  2013.

\bibitem{BCLS2018}
Alin Bostan, Fr{\'e}d{\'e}ric Chyzak, Pierre Lairez, and Bruno Salvy.
\newblock Generalized {H}ermite reduction, creative telescoping and definite
  integration of {D}-finite functions.
\newblock In {\em {P}roceedings of {ISSAC}'18}, pages 95--102. ACM, New York,
  2018.

\bibitem{BLS2013}
Alin Bostan, Pierre Lairez, and Bruno Salvy.
\newblock Creative telescoping for rational functions using the
  {G}riffiths-{D}work method.
\newblock In {\em {P}roceedings of {ISSAC}'13}, pages 93--100. ACM, New York,
  2013.

\bibitem{BLS2017}
Alin Bostan, Pierre Lairez, and Bruno Salvy.
\newblock Multiple binomial sums.
\newblock {\em J. Symbolic Comput.}, 80(part 2):351--386, 2017.

\bibitem{CvHL2010}
Yongjae Cha, Mark \Hoeij{van Hoeij}, and Giles Levy.
\newblock Solving recurrence relations using local invariants.
\newblock In {\em Proceedings of {ISSAC}'10}, pages 303--309. ACM, New York,
  2010.

\bibitem{CDWZ2021}
Shaoshi Chen, Lixin Du, Rong-Hua Wang, and Chaochao Zhu.
\newblock On the existence of telescopers for rational functions in three
  variables.
\newblock {\em J. Symbolic Comput.}, 104:494--522, 2021.

\bibitem{CDZ2019}
Shaoshi Chen, Lixin Du, and Chaochao Zhu.
\newblock Existence problem of telescopers for rational functions in three
  variables: the mixed cases.
\newblock In {\em Proceedings of {ISSAC}'19}, pages 82--89. ACM, New York,
  2019.

\bibitem{CvHKK2018}
Shaoshi Chen, Mark \Hoeij{van Hoeij}, Manuel Kauers, and Christoph Koutschan.
\newblock Reduction-based creative telescoping for fuchsian {D}-finite
  functions.
\newblock {\em J. Symbolic Comput.}, 85:108--127, 2018.

\bibitem{CHLW2016}
Shaoshi Chen, Qing-Hu Hou, George Labahn, and Rong-Hua Wang.
\newblock Existence problem of telescopers: beyond the bivariate case.
\newblock In {\em Proceedings of {ISSAC}'16}, pages 167--174. ACM, New York,
  2016.

\bibitem{CHKL2015}
Shaoshi Chen, Hui Huang, Manuel Kauers, and Ziming Li.
\newblock A modified {A}bramov-{P}etkov\v sek reduction and creative
  telescoping for hypergeometric terms.
\newblock In {\em Proceedings of {ISSAC}'15}, pages 117--124. ACM, New York,
  2015.

\bibitem{CKK2016}
Shaoshi Chen, Manuel Kauers, and Christoph Koutschan.
\newblock Reduction-based creative telescoping for algebraic functions.
\newblock In {\em Proceedings of {ISSAC}'16}, pages 175--182. ACM, New York,
  2016.

\bibitem{ChSi2014}
Shaoshi Chen and Michael~F. Singer.
\newblock On the summability of bivariate rational functions.
\newblock {\em J. Algebra}, 409:320--343, 2014.

\bibitem{CHM2006}
William Y.~C. Chen, Qing-Hu Hou, and Yan-Ping Mu.
\newblock A telescoping method for double summations.
\newblock {\em J. Comput. Appl. Math.}, 196(2):553--566, 2006.

\bibitem{Chyz2000}
Fr{\'e}d{\'e}ric Chyzak.
\newblock An extension of {Z}eilberger's fast algorithm to general holonomic
  functions.
\newblock {\em Discrete Math.}, 217(1-3):115--134, 2000.
\newblock Formal power series and algebraic combinatorics (Vienna, 1997).

\bibitem{ChSa1998}
Fr{\'e}d{\'e}ric Chyzak and Bruno Salvy.
\newblock Non-commutative elimination in {O}re algebras proves multivariate
  identities.
\newblock {\em J. Symbolic Comput.}, 26(2):187--227, 1998.

\bibitem{ClvH2006}
Thomas Cluzeau and Mark \Hoeij{van Hoeij}.
\newblock Computing hypergeometric solutions of linear recurrence equations.
\newblock {\em Appl. Algebra Engrg. Comm. Comput.}, 17(2):83--115, 2006.

\bibitem{GGSZ2003}
J{\"u}rgen Gerhard, Mark Giesbrecht, Arne Storjohann, and Eugene~V. Zima.
\newblock Shiftless decomposition and polynomial-time rational summation.
\newblock In {\em Proceedings of {ISSAC}'03}, pages 119--126. ACM, New York,
  2003.

\bibitem{GHLZ2019}
Mark Giesbrecht, Hui Huang, George Labahn, and Eugene Zima.
\newblock Efficient integer-linear decomposition of multivariate polynomials.
\newblock In {\em Proceedings of {ISSAC}'19}, pages 171--178. ACM, New York,
  2019.

\bibitem{vHoe2017}
Mark \Hoeij{van Hoeij}.
\newblock Closed form solutions for linear differential and difference
  equations.
\newblock In {\em Proceedings of {ISSAC}'17}, pages 3--4. ACM, New York, 2016.

\bibitem{vdHo2021}
Joris \Hoeven{van der Hoeven}.
\newblock Constructing reductions for creative telescoping: the general
  differentially finite case.
\newblock {\em Appl. Algebra Engrg. Comm. Comput.}, 32(5):575--602, 2021.

\bibitem{HoWa2015}
Qing-Hu Hou and Rong-Hua Wang.
\newblock An algorithm for deciding the summability of bivariate rational
  functions.
\newblock {\em Adv. in Appl. Math.}, 64:31--49, 2015.

\bibitem{Huan2016}
Hui Huang.
\newblock New bounds for hypergeometric creative telescoping.
\newblock In {\em Proceedings of {ISSAC}'16}, pages 279--286. ACM, New York,
  2016.

\bibitem{Kout2009}
Christoph Koutschan.
\newblock {\em {Advanced applications of the holonomic systems approach}}.
\newblock PhD thesis, RISC-Linz, Johannes Kepler University, September 2009.

\bibitem{Kout2010a}
Christoph Koutschan.
\newblock A fast approach to creative telescoping.
\newblock {\em Math. Comput. Sci.}, 4(2-3):259--266, 2010.

\bibitem{Kout2010b}
Christoph Koutschan.
\newblock Holonomic{F}unctions (user's guide), 2010.
\newblock Technical Report 10-01, RISC Report Series, Johannes Kepler
  University, Linz, Austria.

\bibitem{Le2003a}
Ha~Q. Le.
\newblock A direct algorithm to construct the minimal {$Z$}-pairs for rational
  functions.
\newblock {\em Adv. in Appl. Math.}, 30(1-2):137--159, 2003.

\bibitem{LiZh2013}
Ziming Li and Yi~Zhang.
\newblock An algorithm for decomposing multivariate hypergeometric terms, 2013.
\newblock A contributed talk in CM'13.

\bibitem{LPR2002}
Russell Lyons, Peter Paule, and Axel Riese.
\newblock A computer proof of a series evaluation in terms of harmonic numbers.
\newblock {\em Appl. Algebra Engrg. Comm. Comput.}, 13(4):327--333, 2002.

\bibitem{Paul1995}
Peter Paule.
\newblock Greatest factorial factorization and symbolic summation.
\newblock {\em J. Symbolic Comput.}, 20(3):235--268, 1995.

\bibitem{Petk1992}
Marko Petkov{\v{s}}ek.
\newblock Hypergeometric solutions of linear recurrences with polynomial
  coefficients.
\newblock {\em J. Symbolic Comput.}, 14(2-3):243--264, 1992.

\bibitem{Schn2007}
Carsten Schneider.
\newblock Simplifying sums in {$\Pi\Sigma^*$}-extensions.
\newblock {\em J. Algebra Appl.}, 6(3):415--441, 2007.

\bibitem{Taka1990}
Nobuki Takayama.
\newblock An algorithm of constructing the integral of a module --an infinite
  dimensional analog of {G}r\"{o}bner basis.
\newblock In {\em Proceedings of {ISSAC}'90}, pages 206--211. ACM, New York,
  1990.

\bibitem{Wegs1997}
Kurt Wegschaider.
\newblock Computer {G}enerated {P}roofs of {B}inomial {M}ulti-{S}um
  {I}dentities.
\newblock Diploma Thesis, RISC, J. Kepler University, Linz, May 1997.

\bibitem{WiZe1992a}
Herbert~S. Wilf and Doron Zeilberger.
\newblock An algorithmic proof theory for hypergeometric (ordinary and
  ``{$q$}'') multisum/integral identities.
\newblock {\em Invent. Math.}, 108(3):575--633, 1992.

\bibitem{Zeil1990b}
Doron Zeilberger.
\newblock A holonomic systems approach to special functions identities.
\newblock {\em J. Comput. Appl. Math.}, 32(3):321--368, 1990.

\bibitem{Zeil1991}
Doron Zeilberger.
\newblock The method of creative telescoping.
\newblock {\em J. Symbolic Comput.}, 11(3):195--204, 1991.

\end{thebibliography}

\newcommand{\Gathen}{\relax}\newcommand{\Hoeij}{\relax}\newcommand{\Hoeven}{\relax}\def\cprime{$'$}
  \def\cprime{$'$} \def\cprime{$'$} \def\cprime{$'$} \def\cprime{$'$}
  \def\cprime{$'$} \def\cprime{$'$} \def\cprime{$'$} \def\cprime{$'$}
  \def\polhk#1{\setbox0=\hbox{#1}{\ooalign{\hidewidth
  \lower1.5ex\hbox{`}\hidewidth\crcr\unhbox0}}} \def\cprime{$'$}

\end{document}